\documentclass[11pt]{amsart}

\usepackage{amssymb}
\usepackage{amsmath}
\usepackage{textcomp}
\usepackage{graphicx}
\usepackage{url}
\usepackage{listings}
\usepackage{bm}
\usepackage{authblk}
\usepackage{booktabs}
\usepackage{multirow}
\usepackage[margin=0.8in]{geometry}
\usepackage{xcolor}
\usepackage[norelsize]{algorithm2e}
\usepackage{accents}
\newlength{\dhatheight}
\newcommand{\dhat}[1]{%
    \settoheight{\dhatheight}{\ensuremath{\hat{#1}}}%
    \addtolength{\dhatheight}{-0.35ex}%
    \hat{\vphantom{\rule{1pt}{\dhatheight}}%
    \smash{\hat{#1}}}}

\SetCommentSty{mycommfont}

\newtheorem{thm}{Theorem}[section]

\newtheorem{cor}{Corollary}[section]

\newtheorem{exm}{Example}[section]
\newtheorem{dfn}{Definition}[section]
\newtheorem{rem}{Remark}[section]

\begin{document}

\makeatletter
\g@addto@macro{\endabstract}{\@setabstract}
\newcommand{\authorfootnotes}{\renewcommand\thefootnote{\@fnsymbol\c@footnote}}%
\makeatother

\begin{center}

\LARGE
  ABCDepth: efficient algorithm \\[.2em] for Tukey depth \par \bigskip

  \normalsize
  \authorfootnotes
  Milica Bogi\'cevi\'c\footnote{antomripmuk@yahoo.com} and Milan Merkle\footnote{emerkle@etf.rs} \par \bigskip

  {\em University of Belgrade, Faculty of Electrical Engineering, Bulevar Kralja Aleksandra 73, 11120 Belgrade, Serbia} \par \bigskip

\end{center}

\medskip

{\small
{\bf Abstract.} We present a new fast approximate algorithm for Tukey (halfspace) depth level sets and its implementation. Given a $d$-dimensional data set
for any $d\geq 2$, the algorithm is based on a representation of level sets as intersections of balls
in ${\bf R}^d$ (M. Merkle, J. Math. Anal. Appl. {\bf 370} (2010)). Our approach does not need calculations of projections of sample points to directions.
This novel idea enables calculations of level sets in very high dimensions with complexity which is linear in $d$, which provides a great advantage
over all other approximate algorithms. Using different versions of this algorithm we demonstrate approximate calculations of the deepest set of
points ("Tukey median"),
Tukey's depth of a sample point and of out-of-sample
point as well as approximate level sets that can be used for constructing depth contours, all with a linear in $d$ complexity. An additional theoretical advantage of this approach is that the
data points are not assumed to be
in "general position".  Examples with real and synthetic data show that the  executing time of the
 algorithm in all mentioned
versions in high dimensions is much smaller  than other implemented algorithms and that it can accept thousands of multidimensional observations.

 }

\medskip

\section{Introduction}
\label{seci}

\medskip

A basic statistical task is to simplify a large amount of data using some values derived from the data set as representative points.
Among many ways to choose representative points, a  natural idea is to choose those that are located in the center of the data set.
One way to define a center is to define what is meant by deepness, and then to define the center as the set of deepest points.

 Although this paper is about multivariate medians and related notions, for completeness and understanding some ideas, we start from the univariate case.
Talking in terms of probability distributions, let $X$ be a random variable and let $\mu=\mu_X$ be the corresponding distribution, i.e., a probability measure
on $(\bf R,\mathcal B)$ so that $P(X\leq x)= \mu \{ (-\infty, x ]\}$.  For univariate case, a median of $X$  (or a median of $\mu_X$) is any number $m$ such that
$P(X\leq m)\geq 1/2$ and $P(X\geq m) \geq 1/2$.  In terms of data set, this property means that to reach any median point from outside of the data set,
we have to pass at least $1/2$ of data points, so this is the deepest point within the data set. With respect to this definition, we can define the depth of
any point $x\in \bf R$ as
\begin{equation}
\label{depd1}
D(x,\mu)=\min\{P(X\leq x), P(X\geq x)\} =\min \{\mu ((-\infty, x]), \mu ([x, +\infty))\}.
\end{equation}
The set of all median points $\{\rm Med \mu\}$ is  a non-empty compact interval (can be a singleton).
It can be shown that (see \cite{merkle05}, \cite{merkle10})
\begin{equation} \label{unimed}
\{\rm Med \mu\} =  \bigcap_{J=[a,b]:\ \mu (J) >1/2} J,
\end{equation}
 and (\ref{unimed}) can be taken for an alternative (equivalent) definition of univariate median set.  In ${\bf R}^d$ with $d>1$, there are quite a few
 different concepts of depth and medians (see  for example \cite{survey15}, \cite{small90}, \cite{zuoserf00}). In this paper we
 propose an algorithm for  halfspace depth (Tukey's depth, \cite{tukey75}), which is based on extension and generalization of (\ref{unimed}) to ${\bf R}^d$
 with balls in place of  intervals as in  \cite{merkle10}.

 The rest of the paper is organized as follows. Section 2 deals with a   theoretical background of the algorithm in a broad sense.
 In  Section 3 we present  approximate algorithm for finding Tukey median as well as  versions of the same algorithm for finding Tukey depth of a sample point,
the depth of  out-of-sample point,  and for  data contours. We also provide a derivation of complexity for each version of the algorithm and
present examples. Section 4 provides a comparison with several other algorithms in terms of  performances.

\medskip

\section{Theoretical background: Depth functions based on families of convex sets}
\label{sect}

\medskip

\begin{dfn} \label{ddepthf} Let $\mathcal V$ be a family of convex sets in ${\bf R}^d$, $d\geq 1$, such that:  (i) $\mathcal V$ is closed under translations and (ii) for every ball $B\in {\bf R}^d$ there exists a set $V\in \mathcal V$ such that $B\in \mathcal V$.  Let $\mathcal U$ be the collection of complements of sets in $\mathcal V$.  For a given probability measure $\mu$ on ${\bf R}^d$, let us  define
\begin{equation}
 \label{depdg}
D_{\mathcal V}(x;\mu)= \inf \{\mu(U)\; |\; x\in U\in \mathcal U\}= 1-\sup\{\mu(V)\; |\; V\in \mathcal V,\ x\in V'\}
\end{equation}
The function $x\mapsto D(x;\mu,\mathcal V)$ will be called a depth function based on the family $\mathcal V$.
\end{dfn}

\begin{rem}{\rm Definition \ref{ddepthf} is a special case of Type $D$ depth functions as defined in \cite{zuoserf00} which can be obtained by generalizations of (\ref{depd1}) to higher dimensions.
The conditions stated in \cite{merkle10} that provide desirable behavior of the depth function, are satisfied in this special case, with additional requirements
that sets in $\mathcal V$ are closed or compact.}
\end{rem}

\medskip

\begin{exm}\label{Edfc} {\rm $1^{\circ}$ Let $\mathcal V$ be the family of all compact intervals $[a,b]\subset \bf R$. As shown in  \cite{merkle10}, the depth function
based on $\mathcal V$ is the same as the one defined by (\ref{depd1}).

$2^{\circ}$ With $d=2$, consider the family $\mathcal V$ of rectangles with sides parallel to coordinate axes. The corresponding depth function reaches
its maximum $D_{\max} \geq 1/2$ at coordinate-wise median. The same holds for $d>2$, with "boxes" whose sides are parallel to coordinate hyper-planes.

$3^{\circ}$ For $d>1$, let   $K$ be a closed convex cone in ${\bf R}^d$, with vertex at origin, and suppose that there exists a closed hyperplane $\pi$, such
that $\pi \cap K =\{0\}$ (that is, $K\setminus\{0\}$ is a subset of one of open halfspaces determined by $\pi$). Define a relation $\preceq$ by
$\boldsymbol x \preceq \boldsymbol y \iff \boldsymbol y -\boldsymbol x \in K$. Generalized intervals based on this partial order can be defined as
\[ [\boldsymbol a,\boldsymbol b]= \{ \boldsymbol x\; |\; \boldsymbol x-\boldsymbol a \in K \wedge \boldsymbol b-\boldsymbol x \in K\} = (\boldsymbol a +K) \cap (\boldsymbol b -K).\] Now let us take $\mathcal V$ to be a collection of all such
intervals with finite endpoints and define the depth by (\ref{depdg}). It can be shown (\cite[Section 3]{merkle10}) that the maximal depth is always
$\geq 1/2$, and the median set can be found using formula
(\ref{unimed}) with generalized intervals.

$4^{\circ}$  Let us consider a family $\mathcal V$ of all closed halfspaces in ${\bf R}^2$. Let $X$ be a random point in ${\bf R}^2$ with $P(X=A)=P(X=B)=P(X=C)=1/3$,
where $ABC$ is a non-degenerated triangle. Here all points inside and on the border of the triangle $ABC$ have the depth $1/3$ and the depth of other points
is equal to zero. Similar examples can be made for arbitrary dimension (see Example 4.1. in \cite{merkle10}).}

\end{exm}

From the above examples we see that
\begin{itemize}
\item A family $\mathcal V$ is not uniquely determined by the depth function; if we start with different collections $\mathcal V_1$ and $\mathcal V_2$,
the corresponding depth functions based on them can be the same (see \cite[Theorem 4.2]{merkle10} for a set of sufficient conditions).
\item The maximal depth in $d>1$ doesn't need to be $1/2$ as in the scalar case. The following general result
(\cite[Theorem 4.1]{merkle10}, see also particular case for Tukey's depth in \cite[Proposition 9]{rouru99}) says that the maximal depth has to be $1/(d+1)$ or bigger:
 \end{itemize}

\begin{thm} \label{alpha} Let $\mathcal V$ be any non-empty family of compact convex subsets of ${\bf R}^d$ satisfying the conditions as in Definition \ref{ddepthf}.
Then for any
probability measure $\mu$ on ${\bf R}^d$ there exists a point $x\in {\bf R}^d$ such that $D_{\mathcal V}(x;\mu)\geq \frac{1}{d+1}$.
\end{thm}

The set of points with maximal depth is called {\em the center of distribution} and denoted  as $ C(\mu,\mathcal V)$.
In general, one can observe {\em level sets} (or {\em depth regions} or {\em depth-trimmed regions}) of level  $\alpha$ which are defined by
\begin{equation}
 \label{salpha}
 S_{\alpha}= S_{\alpha}(\mu,\mathcal V):=\{ x\in {\bf R}^d\; | \; D(x; \mu,\mathcal V)\geq \alpha \} .
\end{equation}

Clearly, if $\alpha_1<\alpha_2$ then $S_{\alpha_1} \supseteq S_{\alpha_2}$ and $S_{\alpha}=\emptyset$ for $\alpha>\alpha_{m}$, where $\alpha_m$ is
the maximal depth for given probability measure $\mu$.

The borders of depth level sets  are called {\em depth contours} (in two dimensions)  or {\em depth surfaces} in general. Let us note that the
all statistical inference based on multidimensional depths is performed using level sets and contours (see \cite{doga92,dughocha11,zhouserf08,rouru99}),
and that it is rarely necessary to find a depth of a particular point. On the other hand, in order to describe level sets and the center of
distribution we do not need to calculate depth functions, as the next result shows (\cite{zuoserf00} and Theorem 2.2. in \cite{merkle10}).

\begin{thm} \label{repls}
Let $D(x;\mu,\mathcal V)$ be defined for $x\in {\bf R}^d$ as in Definition \ref{ddepthf}. Then for any $\alpha\in (0,1]$
\begin{equation}
\label{lsvi}
S_{\alpha}(\mu,\mathcal V) =
\bigcap_{V\in \mathcal V, \mu(V)>1-\alpha} V.
\end{equation}
\end{thm}

The center of a distribution is then the smallest non-empty level set; equivalently,
\begin{equation}
\label{cenf}
 C(\mu,\mathcal V)=\bigcap_{\alpha:S_\alpha \neq \emptyset} S_{\alpha} (\mu,\mathcal V)
\end{equation}

Since sets in $\mathcal V$ are convex, the level sets are also convex.

From (\ref{salpha}) and (\ref{lsvi}) we can see that the depth function can be uniquely reconstructed starting
from level sets.

\begin{cor} \label{deviale}For given $\mu$ and $\mathcal V$, let $S_{\alpha}$, $\alpha\geq 0$ be defined as in (\ref{lsvi}),
with $S_{\alpha} =\emptyset$ for $\alpha >1$.  Then the function $D: {\bf R}^d \mapsto [0,1]$
defined by
\begin{equation}
\label{deplf}
D(x) = h\quad \iff \quad x\in S_{\alpha}, \alpha \leq h \quad {\rm and}\quad x\not\in S_{\alpha}, \alpha >h
\end{equation}
is the unique depth function  such that (\ref{salpha}) holds.
\end{cor}

The algorithm that we propose in this paper primary finds  level sets based on the formula (\ref{lsvi}), rather then directly the depth of particular points.
The depth of a single point, if needed, can be calculated via Corollary \ref{deviale}. The algorithm will be demonstrated in the case of half-space depth,
which is described in the next section.

\medskip

\section{ABCDepth Algorithm for Tukey depth: Implementation and the Output}
\label{seca}

\medskip
The most popular choice among depth functions of Definition \ref{ddepthf} is the one which is based on half-spaces,
also called Tukey's depth \cite{tukey75}. Here $\mathcal V$ is the family of all open half-spaces, and the complements are
closed half-spaces, so the usual definition of Tukey depth is obtained from (\ref{depdg}) as

\begin{equation}
\label{odtd}
 D(x;\mu)=  \inf \{\mu(H)\; |\; x\in H\in \mathcal H\},
\end{equation}

where $\mathcal H$ is the family of all closed halfspaces. In this section, we consider only half-space depth, so we use the notation $D(x,\mu)$ instead of
$D_{\mathcal V}(x,\mu)$.
As  already noticed in Section \ref{sect}, a depth function can be defined based on different families $\mathcal V$.
We say that families $\mathcal V_1$ and $\mathcal V_2$ are
depth-equivalent if $D_{\mathcal V_1}(x;\mu)=D_{\mathcal V_2}(x;\mu)$ for all $x\in {\bf R}^d$ and all probability measures $\mu$.
Sufficient conditions for depth-equivalence are given in  \cite[Theorem 2.1]{merkle10}, and it was shown there that in the case of
half-space depth the following families are depth-equivalent:
a) Family of all open halfspaces; b) all closed  halfspaces; c) all convex sets; d) all compact convex sets; e) all closed or open balls.

For determining level sets we choose closed balls,  and so we can define $\mathcal V$ as a set of all closed balls (hyper-spheres)
and (exact) level sets can be found as
\begin{equation}
\label{sabs}
S_{\alpha}(\mu,\mathcal V) =
\bigcap_{B\in \mathcal V, \mu(B)>1-\alpha} B.
\end{equation}

From now on we consider only the case when the underlying probability measure $\mu$ is derived from a given data set.

\subsection{The sample version}\label{saver}
In the setup with data sets,  we have a sample of $n$ points $\{x_1,\ldots, x_n\}$ (with repetitions allowed) and we may
use the counting measure defined as
\begin{equation}
\label{cmes}
 \mu (A)= \frac{ \# \{ x_i: \ x_i \in A\}}{n}.
\end{equation}
In this part we assume to have a fixed sample of $n$ points, so we don't need to explicitly acknowledge the dependence of sample and its cardinality.
The level sets in (\ref{sabs}) for $\alpha\in (0,1]$ can be found as

\begin{equation}
\label{sabsds}
S_{\alpha} = \bigcap_{B\in \mathcal V, \#\{x_i:\ x_i \in B \} \geq \lfloor n(1 - \alpha) + 1 \rfloor} B,
\end{equation}
where we write $S_{\alpha}$ instead of $S_{\alpha}(\mu,\mathcal V)$, assuming that $\mathcal V$ is the collection of all closed balls and $\mu$ is defined
as in (\ref{cmes}).  In a practical realization, we start with a finite collection   $\mathcal V_N\subset \mathcal V$, of $N$ balls $B_i$, $i=1,\ldots, N$
 which contain at least $\lfloor n(1 - \alpha) + 1 \rfloor $ points. In is natural to assume that if we want more than  $N$ balls, then we just add
new ones to the collection $\mathcal V_N$, i.e,
\begin{equation}
 \label{NN}
 \mathcal V_{N_1}\subset \mathcal V_{N_2}\qquad \mbox{for}\ N_1<N_2.
\end{equation}

Now for fixed $N$ and $\mathcal V_N$, we intersect balls  $B_i \in \mathcal V_N$ one by one, so that the $k$-the step we have
the approximate level set

\begin{equation}
\label{alevs}
\hat{S}_{\alpha, k} := \cap_{i=1}^{k}B_i, 1\leq k\leq N.
\end{equation}

From the assumption (\ref{NN}) it follows that

\begin{equation}
\label{aproxN}
\hat{S}_{\alpha,N} \supseteq \hat{S}_{\alpha,N+1} \supseteq \cdots \supseteq S_{\alpha}.
\end{equation}

In general, for a given $\alpha$  it can happen that
$\hat{S}_{\alpha, k}=\emptyset$. Since it is computationally
hard problem to determine whether or not  the set
(\ref{alevs}) is empty, and also for the purpose of
visualization, we  need to have some points inside the balls to
decide if they belong to $\hat{S}_{\alpha, k}$ or not. Let
$C_M$ be a discrete set of points in ${\bf R}^d$, $|C_M|=M\geq N$,
such that (i) each ball $B_i$, $i=1,\ldots N$ of (\ref{alevs})
contains at least one $c\in C$ and (ii) each point $c\in C$
belongs to at least one ball $B_i$. It is natural to assume
that sets $C_M$ are increasing with $M$, that is,
\begin{equation}
\label{ranptn}
 M_1 < M_2 \implies C_{M_1} \subset C_{M_2}
\end{equation}
Let us define
\begin{equation} \label{setlevp}
\dhat{S}_{\alpha, k, M} := \{ c\in C_M \; | \; c\in \hat{S}_{\alpha, k} \},\qquad \dhat{S}_{\alpha, N, M}= \bigcap_{k=1}^N \dhat{S}_{\alpha, k, M}.
\end{equation}

For a fixed $N$  and every $M\geq N$ we have that

\begin{equation} \label{relbetwd}
\dhat{S}_{\alpha,N,M}\subseteq \dhat{S}_{\alpha,N,M+1}  \subset \hat{S}_{\alpha,N},
\end{equation}
and also for every $k<N$

\begin{equation} \label{relbetdwk}
\dhat{S}_{\alpha,k,M}\supseteq \dhat{S}_{\alpha,k+1,M}.
\end{equation}

 The process of finding $\dhat{S}_{\alpha,N,M} $ as defined by (\ref{setlevp}) can end at the step $K\leq N$ in two ways:
 (i) if  $\dhat{S}_{\alpha, K, M}=\emptyset$,
 (ii) if $K=N$ or  $\dhat{S}_{\alpha, k, M}$ remains the same non-empty set for all $k$ such that
  $K\leq k \leq N$. In the
case (i) we conclude that $\dhat{S}_{\alpha, N, M}= \emptyset$. In the case (ii) we  have that $\dhat{S}_{\alpha, N, M} = \dhat{S}_{\alpha, K,M}$ and
we accept $\dhat{S}_{\alpha,K,M}$  as an approximation to $S_{\alpha}$ defined
by (\ref{sabsds}).

\begin{rem}  {\rm
$1^{\circ}$ The output of the described procedure is $\dhat{S}_{\alpha}$. In order to make a  contour we can find a convex hull of $\dhat{S}_{\alpha}$
using \textit{QuickHull} algorithm, for example. The relations (\ref{relbetwd}) remain true if $\dhat{S}_{\alpha}$ is replaced with its convex hull.

$2^{\circ}$ Computational experiments indicate that the convex hull of $\dhat{S}_{\alpha, M,N}$ converges to exact $S_{\alpha}$ as $M,N \rightarrow \infty$.
 The proof of that statement would follow from (\ref{aproxN}), (\ref{setlevp}) and (\ref{relbetwd}) under some additional assumptions, which
 will not be further
 elaborated   in this paper.

  $3^{\circ}$ The simplest way to implement the above procedure is to take $M=N=n$, and $C=\{x_1,\ldots, x_n\}$ where
$x_i $ are sample points and intersectional balls are centered at $x_i$. In some cases  $M=N> n$ is needed. In the rest of the
paper we consider only the case $M=N$. For simplicity, in the rest of the paper we will use notation $S_{\alpha}$ in the meaning of $\dhat{S}_{\alpha}$
unless explicitly noted otherwise.
}
\qed
\end{rem}

\subsection{Sample augmented with artificial data points} \label{sectartd}
If we have a large and dense sample of the size $n$, we can take $N=n$ and use balls centered in sample points.
The basic Algorithm 1 of the subsection \ref{subsectm} is presented in that setup.
However, setting $N=n$  may not be sufficient in estimation of level sets and especially
the deepest points (Tukey's median).
As an example, consider a uniform distribution in the region bounded by circles  $x^2+y^2=r^2_i$,
$r_1=1$ and $r_2=2$. It is easy to prove (see also \cite{dughocha11}) that the depth monotonically increases from $0$ outside of the
larger circle, to $1/2$ at the origin, which is the true and unique median.
With a sample from this distribution, we will not have data points inside the inner circle, and we can not identify the median in the way proposed above.

In similar cases and whenever we have sparse data or small sample size $n$, we can still visually identify depth regions and center,
simply by adding artificial points to the data set.
Let the data set contain points $x_1, \ldots, x_n$ and let $x_{n+1},\ldots, x_{N}$ be points chosen from uniform distribution in
some convex domain that contains the whole data set. Then we use a modification of the described procedure in such a way that \label{artdata}
we use augmented data set (all $N$ points) as a criterium for stopping (in cases (i) and (ii) above), but
 $n$ in formulas (\ref{cmes}) and (\ref{sabsds}) is the cardinality of original data set.

Figure \ref{ring} shows the output of ABCDepth algorithm in the
example  described above. By adding artificial data points, we are
able to obtain an approximate
position of the Tukey's  median.

\begin{figure}[!htb]
    \centering
    \includegraphics[width=0.5\textwidth, height=0.4\textheight]{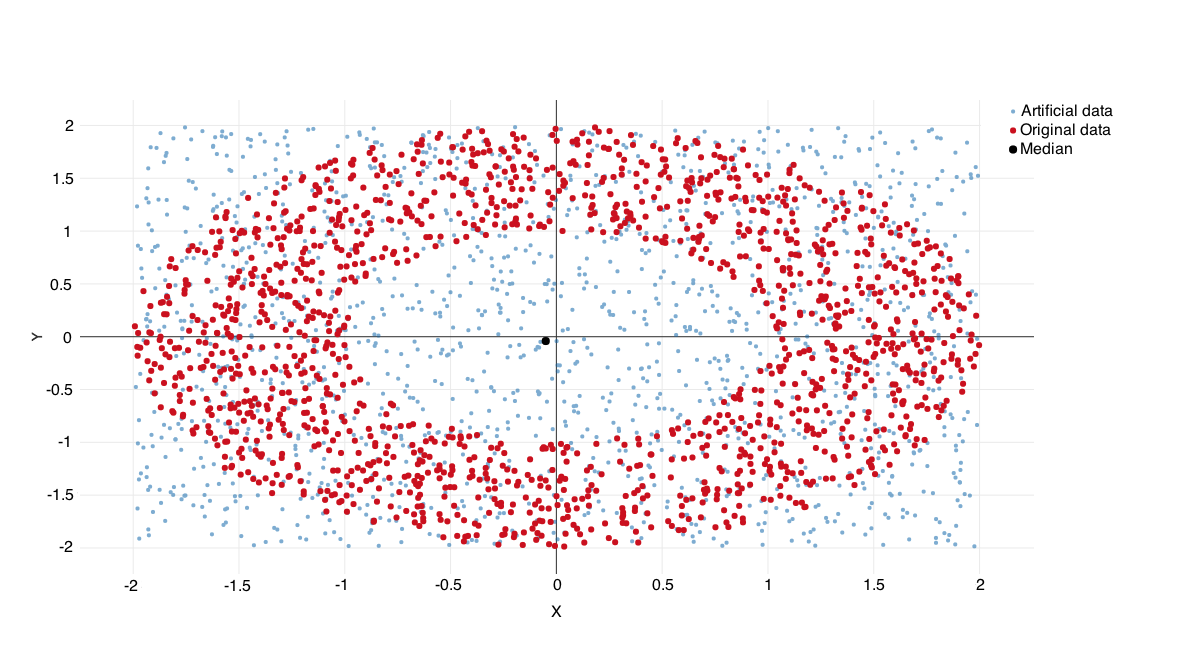}
    \caption{A sample from uniform distribution in a ring (red): Tukey's median (black) found with the aid of artificial points (blue).  }
    \label{ring}
\end{figure}

Let us consider a triangle  as in Example \ref{Edfc}-$4^{\circ}$ of Section \ref{sect} with vertices $A(0, 1)$, $B(-1, 0)$ and $C(1, 0)$.
Assuming that $A,B,C$ are sample points, all points in the interior and on the border of $ABC$ triangle have depth $1/3$,
so the depth reaches its maximum value at $1/3$.
Since the original data set contains only $3$ points, by adding artificial data and applying ABCDepth algorithm we can
visualize  the Tukey's median set as shown in Figure \ref{triangle}.

\begin{figure}[!htb]
    \centering
    \includegraphics[width=0.5\textwidth, height=0.4\textheight]{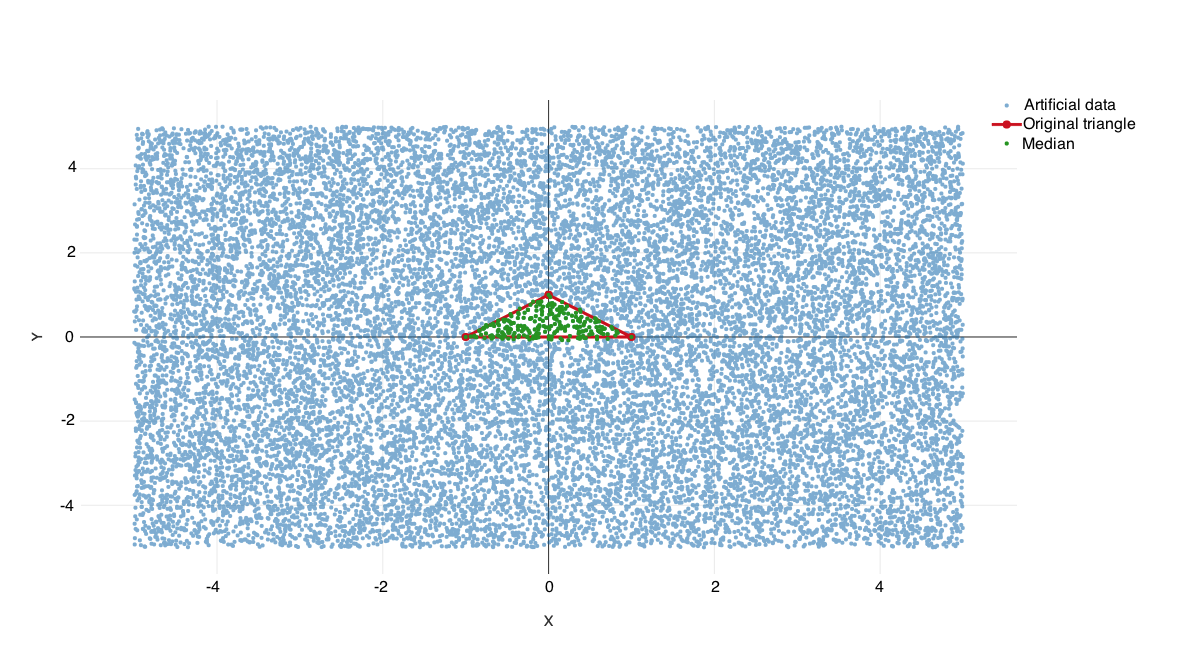}
    \caption{Tukey's median set of red triangle represented as triangle itself and green points inside of the triangle. Blue points
   are artificial data. }
    \label{triangle}
\end{figure}


In the rest of this section, we  describe the details of implementation of the approximate algorithm for finding Tukey's median,
as well as  versions of the same algorithm for finding Tukey depth of a sample point,
the depth of  out-of-sample point,  and for  data contours.

\subsection{Implementation: finding deepest points (Tukey's median)}
\label{subsectm}
In order to execute the calculation in (\ref{sabsds}), the first step is to construct balls for the intersection.
Each ball is defined by its center and contains $\lfloor n(1 - \alpha) + 1 \rfloor$ nearest points, so first we calculate Euclidian inter-distances.
This part of the implementation is described in lines $1-6$ of Algorithm 1.
Distances are stored as a triangular matrix in a \textit{list of lists} structure, where $i$-th list $(i = 1, . . . , n-1)$ contains distances
$d_{i+1,j} , j = 1,\ldots  , i$.
After sorting distances for each point, the structure that contains all $n$ balls is populated (lines 7-10, Algorithm 1).
The structure is represented as  a \textit{hashmap}, where the \textit{key} is a center of a ball, and \textit{value} is a \textit{list} with
$\lfloor n(1 - \alpha) + 1 \rfloor$ nearest points. Now, we intersect balls iteratively by increasing $\alpha$ by $\frac{1}{n}$.

Since this algorithm is meant to  find the deepest location,
 there is no need to start with the minimal value of $\alpha = \frac{1}{n}$; due to Theorem \ref{alpha}, we set the initial value of  $\alpha$ to be
 $\frac{1}{d+1}$.  Balls intersections are shown on Algorithm 1, lines 11-17.


If the input set is sparse  ABCDepth optionally creates an augmented data set of total size $N$ as explained on
page \pageref{artdata} and demonstrated on figures \ref{ring} and  \ref{triangle}). Let $R_1=\{x_1,\ldots, x_n\}$ be the original data set
and let $R_2=\{x_{n+1},\ldots , x_N\}$ be the set of "artificial points".
 The algorithm creates balls with centers  in $R_1\cup R_2$  that contain
 $\lfloor n(1 - \alpha) + 1 \rfloor$ points from  $R_1$. The rest of the algorithm takes three phases we described above.

\begin{algorithm}[H]
  \LinesNumbered
  \SetAlgoLined
  \KwData{Original data, $X_n = (\boldsymbol{x_1}, \boldsymbol{x_1},...,\boldsymbol{x_n}) \in \mathbb{R}^{d \times n}$}
  \KwResult{List of level sets, $S = \{S_{\alpha_1}, S_{\alpha_2},...,S_{\alpha_m}\}$, where $S_{\alpha_m}$ represents a Tukey median}
  \tcc{Note: $S_{\alpha}$ here means $\dhat{S}_{\alpha}$}
  \BlankLine
  \For{$i \gets 2$ \textbf{to} $n$} {
	\For{$j \gets 1$ \textbf{to} $i-1$} {
		 	Calculate Euclidian distance between point $\boldsymbol{x_i}$ and point $\boldsymbol{x_j}$ \;
			Add distance to the list of lists \;
		
	}
  }
   \BlankLine
   \For{$i \gets 1$ \textbf{to} $n$} {
   	Sort distances for point $\boldsymbol{x_i}$  \;
	Populate structure with balls \;
   }
   \BlankLine
   \tcc{Iteration Phase}
   $size = n$, $\alpha_1=\frac{1}{d+1}$, $k=1$ \;
   \While{$size > 1$} {
   	$S_{\alpha_k} = \{ \bigcap_{j}^{n} B_j,  \left | B_j \right | = \lfloor n(1 - \alpha_k) + 1 \rfloor \}$ \;
	$size = \left | S_{\alpha_k} \right |$ \;
	$\alpha_{k+1}=\alpha_k + \frac{1}{n}$ \;
	Add $S_{\alpha_k}$ to $S$ \;
	$k = k + 1$ \;
   }
  \caption{Calculating Tukey median.}
\end{algorithm}

The initial version of ABCDepth algorithm was presented in  \cite{bome15}.

\subsubsection{Complexity}
\begin{thm}
\label{complexitytm}
ABCDepth algorithm for finding approximate Tukey median has order of $O((d + k)n^2 + n^2\log{n})$ time complexity, where $k$
is the number of iterations in the iteration phase.
\end{thm}
\begin{proof}
To prove this theorem we use the pseudocode of  Algorithm 1.
Lines 1-6 calculate Euclidian inter-distances of points. The first \textit{for} loop (line 1) takes all $n$ points, so its complexity is $O(n)$.
Since there is no need calculate $d(x_i, x_i)$ or to calculate $d(x_j, x_i)$ if it is already calculated, the second for loop (line 2) runs in
 $O(\frac{n-1}{2})$ time. Finally, calculation of Euclidian distance takes $O(d)$ time. The overall complexity for lines 1-6 is:
\begin{equation}
\label{one}
O(\frac{nd(n-1)}{2}) \sim O(dn^2)
\end{equation}

Iterating through the \textit{list of lists} obtained in lines 1-6, the first \textit{for} loop (line 7) runs in $O(n)$ time.
For sorting the distances per each point, we use \textit{quicksort} algorithm that takes $O(n \log n)$ comparisons to sort $n$ points
\cite{hoare61}. Structure populating takes $O(1)$ time. Hence, this part of the algorithm has complexity of:
\begin{equation}
\label{two}
O(n^2 \log n).
\end{equation}

In the last phase (lines 11-18), algorithm calculates level sets by intersecting balls constructed in the previous steps.
In every iteration (line 12), all $n$ balls that contain $\lfloor n(1 - \alpha_k) + 1 \rfloor$ are intersected (line 13).
The parameter $k$ can be considered as a number of iterations, i.e. it counts how many times the algorithm enters in \textit{while} loop.
Each intersection has the complexity of $O(\lfloor n(1 - \alpha_k) + 1 \rfloor) \sim O(n)$, due to the property of the hash-based data
structure we use (see for example \cite{fastset}). We can conclude that the iteration phase has complexity of:
\begin{equation}
\label{three}
O(kn^2).
\end{equation}

From (\ref{one}), (\ref{two}) and (\ref{three}),
\begin{equation}
\label{alg1}
O(dn^2) + O(n^2 \log n) + O(kn^2) \sim O((d + k)n^2 + n^2 \log n),
\end{equation}
which ends the proof.

\end{proof}

\begin{rem}{\rm \label{kmaxgp} $1^{\circ}$  From the relations between $S_{\alpha }$, $\hat{S}_{\alpha}$
and $\dhat{S}_{\alpha}$ (in notations as in \ref{saver}, page \pageref{saver}), it follows that the maximal approximative depth of for a given point
can not be greater its than its exact  depth.

$2^{\circ}$ Under the assumption that data points are in the general position, the exact sample maximal depth is $\alpha_m=\frac{m}{n}$, where $m$ is
not greater than $\lceil\frac{n}{2}\rceil$ (see \cite[Proposition 2.3]{doga92}), and so by remark $1^{\circ}$
the number $k$ of steps satisfies the inequality

\begin{equation}
\frac{k-1}{n}\leq  \frac{n+1}{2n} - \frac{1}{d+1},
\end{equation}
and  the asymptotical upper bound for $k$ is $\frac{n}{2}$.}

\end{rem}

\begin{rem} \label{augmtm}{\rm  In the case when we add artificial data points to the original data set,
$n$ in (\ref{alg1}) should be replaced with $N$, where $N$ is the cardinality of the augmented data set.
The upper bound for $k$ in (\ref{alg1}) remains the same. }
\end{rem}

The rates of complexity with respect to $n$ and $d$ of Theorem \ref{complexitytm} are confirmed by simulation results presented in figures
\ref{fig:data_time2} and  \ref{fig:dim_time1}. Measurements are taken on simulated samples of size $n$ from $d$-dimensional
with expectations zero and uncorrelated marginals, with $d \in \{2,...,10\}$ and
\[ n \in \{ 40,80,160,320,640,1280,2560,3000,3500,4000,4500,5000,5500,6000,6500,7000 \}\]
 The results are averaged on $10$ repetitions for each fixed pair $(d,n)$.

\begin{figure}[!htb]
    \centering

    \includegraphics[width=0.7\textwidth]{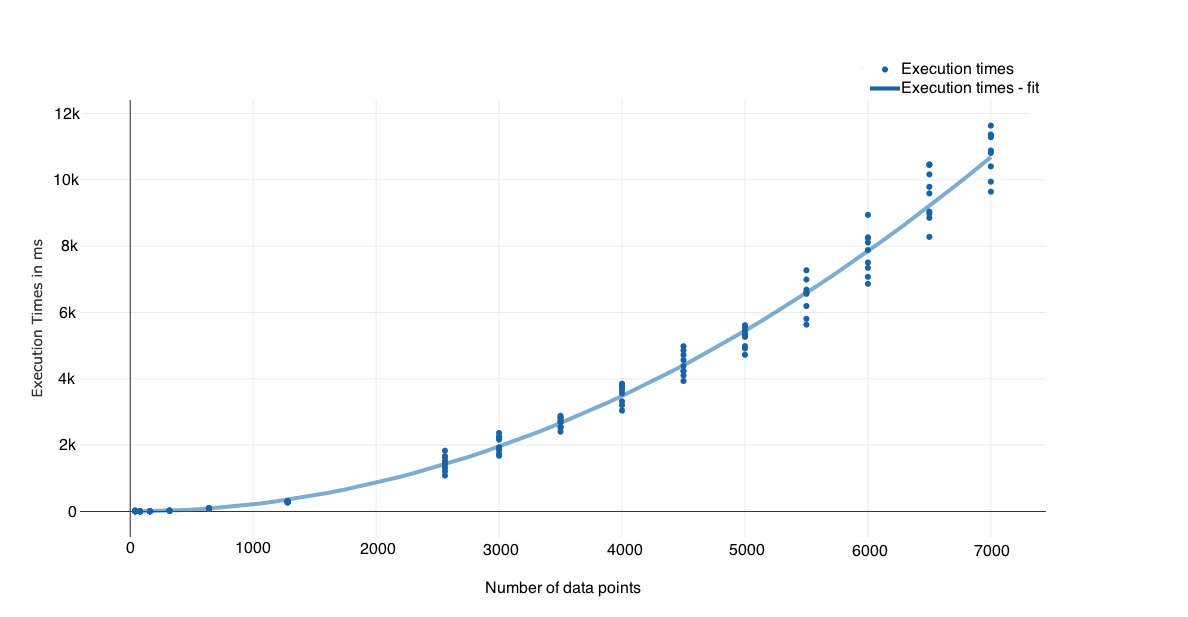}
    \caption{When number  of points increases the execution time grows with the order of $n^2\log n$.}
        \label{fig:data_time2}
\end{figure}

\begin{figure}[!htb]
    \centering
    \includegraphics[width=0.7\textwidth]{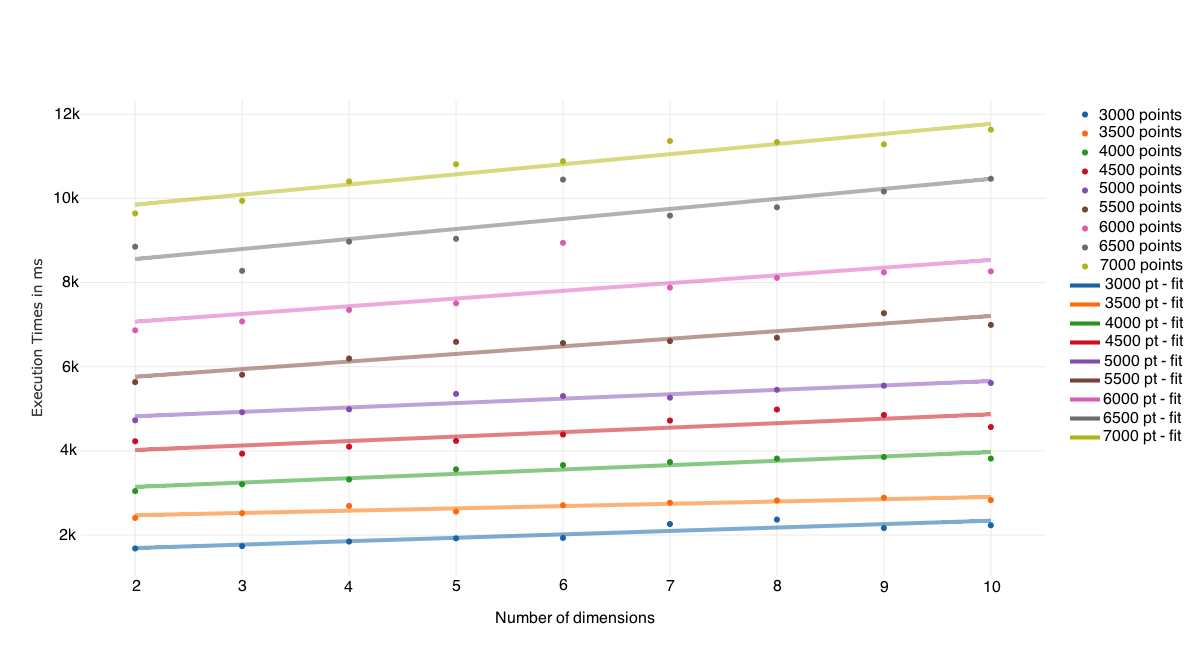}
    \caption{The execution time grows linearly with dimensionality. }
    \label{fig:dim_time1}
\end{figure}

\subsubsection{Examples}
\label{secexamplestm}


Our first example is really simple and it considers $n$ points in dimension $1$ generated from normal $\mathcal N(0,1)$ distribution.
By running ABCDepth in this case with $n=1000$, we get two points (as expected) in the median level set,
$S_{\alpha_0.5}=\{-0.00314, 0.00034\}$. With another sample with  $n=1001$ (odd number)  from
 the same distribution, the median  set is a singleton, $S_{\alpha_0.5} = \{0.0043\}$

Now, we demonstrate data sets generated from bivariate and multivariate normal distribution.

Figure \ref{1000_2}  and Figure \ref{1000_3} show the median calculated from $1000$ points in dimension $2$ and $3$, respectively from normal $\mathcal N(0,1)$ distribution. Starting from $\alpha=\frac{1}{d+1}$ the algorithm produces $\sim 200$ levels sets for $d=2$ and $\sim 300$ level sets for $d=3$, so not all of them are plotted. On both figures the median is represented as a black point with depth $\frac{499}{1000}$ on Figure \ref{1000_2}, i.e. $\frac{493}{1000}$ on Figure \ref{1000_3}.
\begin{figure}[!htb]
    \centering
    \includegraphics[width=0.7\textwidth]{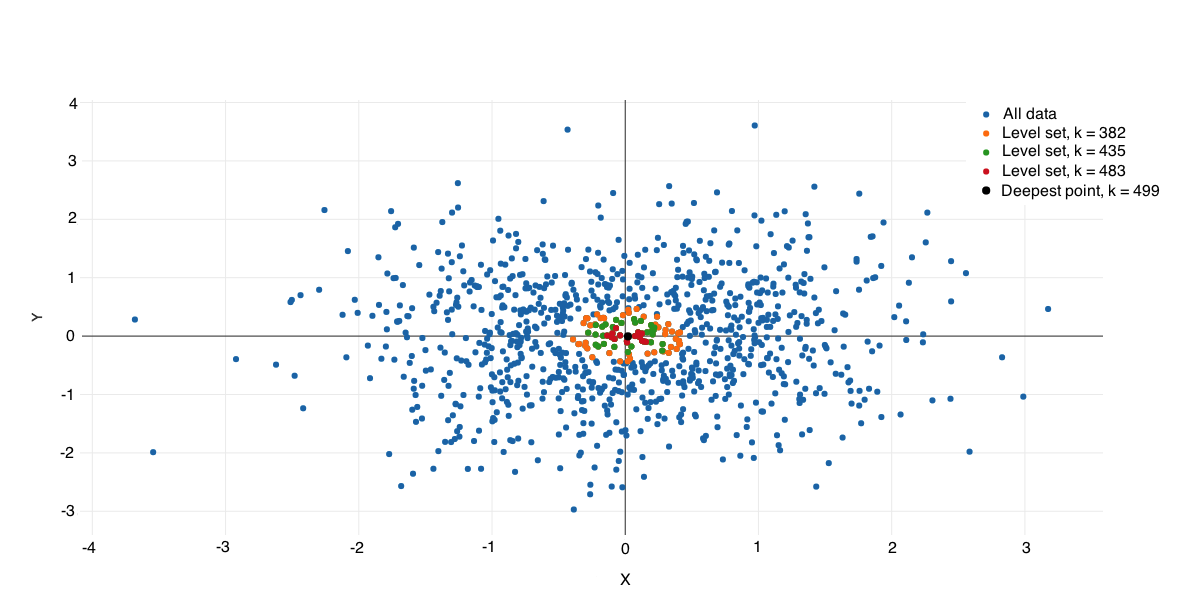}
    \caption{Bivariate normal distribution - four level sets, where the black point at the center is the deepest point.}
    \label{1000_2}
\end{figure}

 \begin{figure}[!htb]
    \centering
    \includegraphics[width=0.7\textwidth]{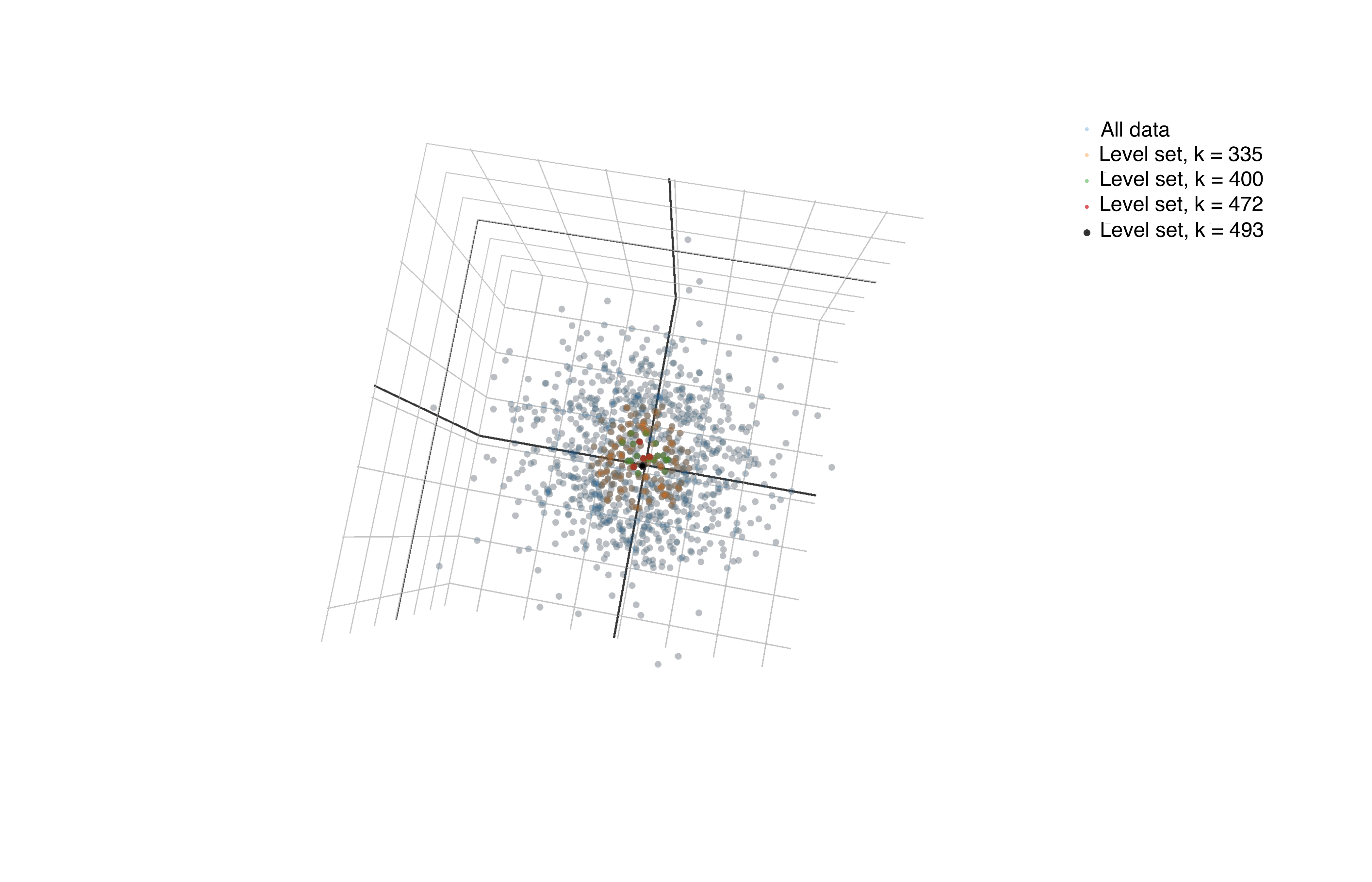}
    \caption{3D normal distribution - four level sets, where the black point at the center is the deepest point.}
    \label{1000_3}
\end{figure}

All data generators that we use in this paper in order to verify and plot the algorithm output were presented at
\cite{bome14} and they are available within  an open source project at \url{https://bitbucket.org/antomripmuk/generators}.

As a real data example, we take a data set which is rather sparse. The data set is taken from \cite{struro00}, and it has been used in several
other papers as a benchmark.  It contains 23 four-dimensional observations in period from 1966 to 1967 that represent seasonally adjusted
changes in auto thefts in New York city. For the sake of clarity, we take only two dimensions:
percent changes in manpower, and seasonally adjusted changes in auto thefts.
The data is downloaded from \url{http://lib.stat.cmu.edu/DASL/Datafiles/nycrimedat.html}.
Figure \ref{single_ny} shows the output of ABCDepth algorithm if we consider only points from the sample (orange point).
Obviously, the approximate median belongs to the original data set.
Then, we run ABCDepth algorithm  with  $1000$ artificial data points from the uniform distribution as  explained in
Section \ref{sect} and earlier in this section.
The approximate median obtained by this run (green point) has the same depth of $\frac{9}{23}$ as the median
calculated using DEEPLOC algorithm \cite{struro00} by running their Fortran code (red point). We check depths of those two points
(green and red) applying \textit{depth} function based on \cite{struyf98} and implemented in R "depth" package \cite{depthr}.
Evidently, the median, in this case, is not a singleton, i.e. there is more than one point with depth $\frac{9}{23}$.
By adding more than $1000$ artificial points, we can get more than one median point. We will discuss this example again in Section \ref{subsecdc}.

 \begin{figure}[!htb]
    \centering
    \includegraphics[width=0.7\textwidth]{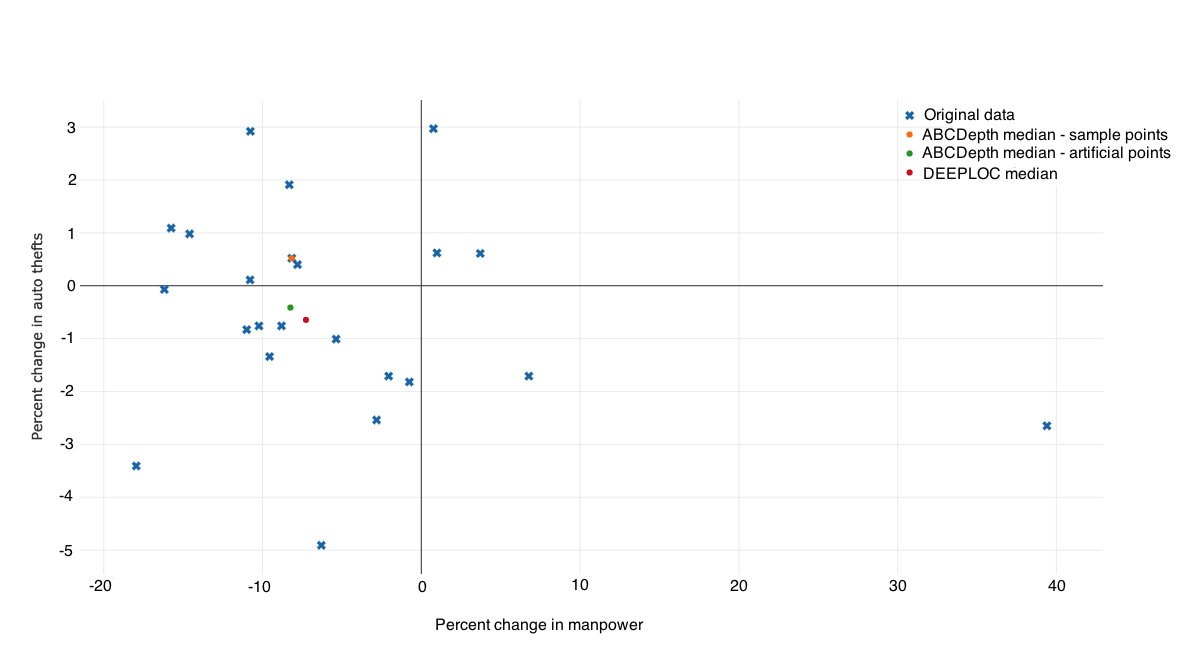}
    \caption{NY crime data set, comparison of Tukey medians using ABCDepth and DEEPLOC.}
    \label{single_ny}
\end{figure}

Another two examples are chosen from \cite{rouru96b}. Figure \ref{artificial_animals} shows $27$ two-dimensional observations
that represent animals brain weight (in g) and the body weight (in kg) taken from \cite{animals}. In order to represent the same
data values, we plotted the logarithms of those measurements as they did in \cite{rouru96b}.

 \begin{figure}[!htb]
    \centering
    \includegraphics[width=0.7\textwidth]{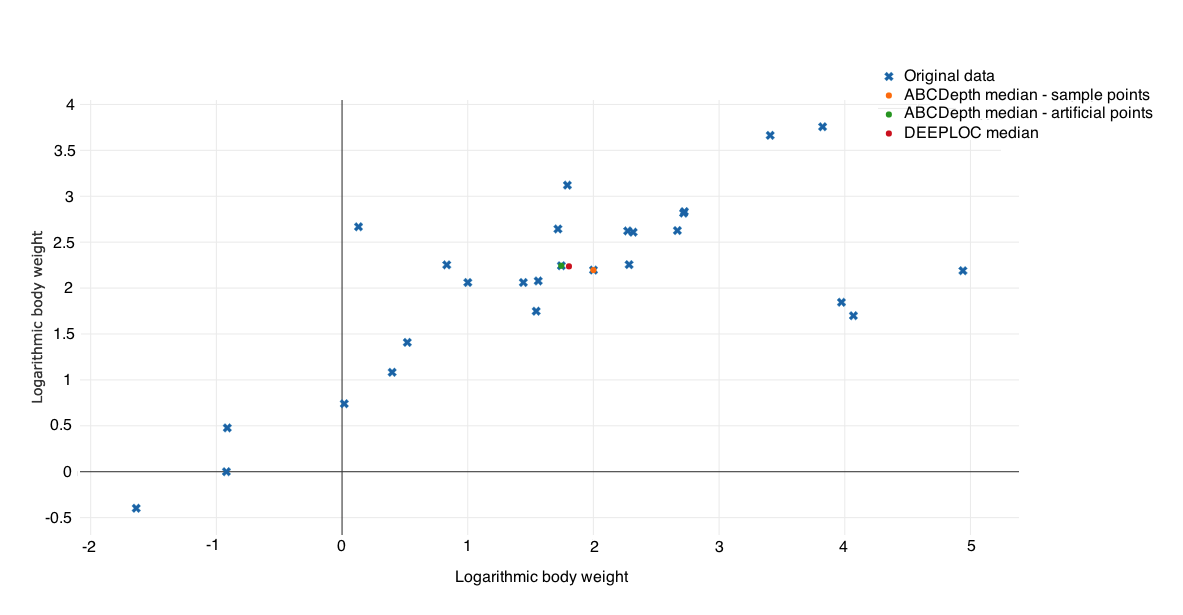}
    \caption{Animals data set, comparison of Tukey medians using ABCDepth and DEEPLOC.}
    \label{artificial_animals}
\end{figure}

Figure \ref{artificial_aircraft} considers the weight and the cost of $23$ single-engine aircraft built between $1947-1979$.
This data set is taken from \cite{aircraft}.

 \begin{figure}[!htb]
    \centering
    \includegraphics[width=0.7\textwidth]{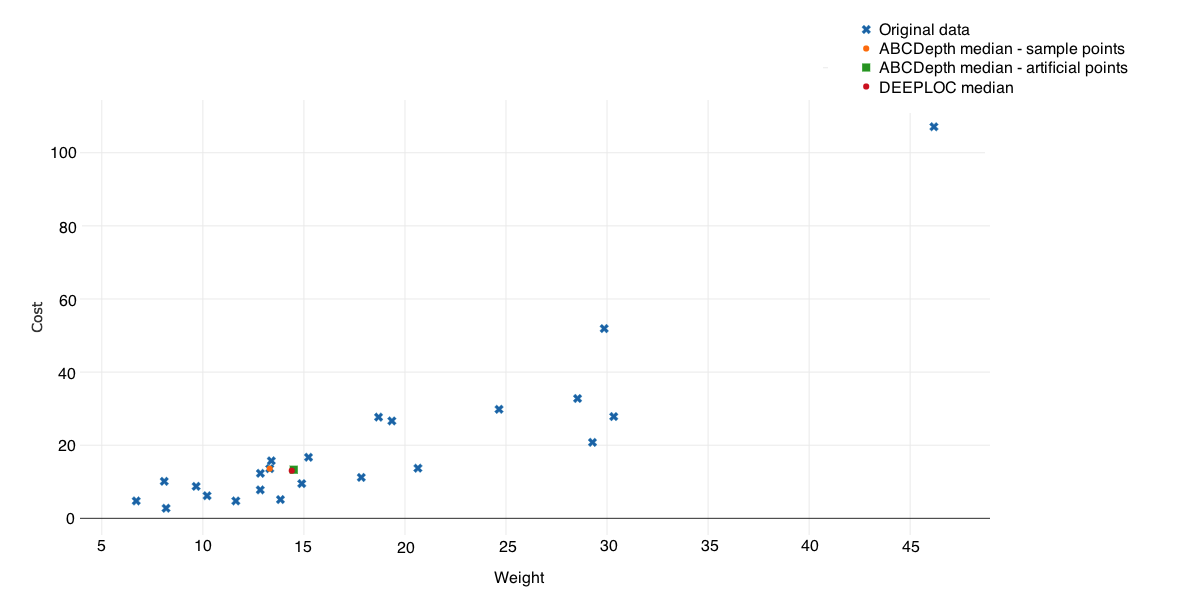}
    \caption{Aircraft data set, comparison of Tukey medians using ABCDepth and DEEPLOC.}
    \label{artificial_aircraft}
\end{figure}

As  in Figure \ref{single_ny}, in those two figures the orange point is the median obtained by running ABCDepth algorithm
using only sample data. Green and red points represent outputs of ABCDepth algorithm applied by adding $1000$ artificial data points from the uniform
 distribution and DEEPLOC median, respectively. These two examples show the importance of out-of-sample points in finding the depth levels and
 Tukey's median.

\subsection{Adapted Implementation: finding the  Tukey's depth of a sample point and out-of-sample point}
\label{subsectd}

Let us recall that by  Corollary \ref{deviale}, a point $\boldsymbol x$ has depth $h$ if and only if $\boldsymbol x \in S_{\alpha}$ for $\alpha \leq h$
and $\boldsymbol x \not \in S_{\alpha} $ for $\alpha > h$. With a sample of size $n$, we can consider only $\alpha =\frac{k}{n}$, $k=1,\ldots, n$, because
for $\frac{k-1}{n} <\alpha< \frac{k}{n}$, we have that $D(\boldsymbol x)\geq \alpha \iff D(\boldsymbol x)\geq \frac{k}{n}$. Therefore, the statement of Corollary
\ref{deviale} adapted to the sample distribution can be formulated as (using the fact that $S_{\beta} \subset S_{\alpha}$ for $\alpha <\beta$):
\begin{equation}
\label{deviales}
D(\boldsymbol x)= \frac{k}{n} \iff \boldsymbol x \in S_{\frac{k}{n}} \quad {\rm and} \quad \boldsymbol x \not\in S_{\frac{k+1}{n}} .
\end{equation}

From (\ref{deviales}) we derive the algorithm for Tukey's depth of a sample point $\bf x$ as follows. Let $\alpha_k = \frac{k}{n}$. The level set
 $S_{\alpha_1}$
contains all points in the sample. Then we construct $S_{\alpha_2}$ as an intersection of $n$ balls that contain $n-1$ sample points.
If $\boldsymbol x \not \in S_{\alpha_2}$, we conclude that $D(\boldsymbol x)=1/n$, and stop. Otherwise, we iterate this procedure till we get the situation as in
right side of (\ref{deviales}), when we conclude that the depth is $\frac{k}{n}$. The output of the algorithm is $k$.

\begin{rem}{\rm As in Remark \ref{kmaxgp}, it can be shown that the approximate depth $k/n$ is never greater than the true depth.

}
\end{rem}

Implementation-wise, in order to improve the algorithm complexity, we do not need to construct the level sets.
It is enough to count balls that contain point $\boldsymbol x$. The algorithm stops when for some $k$ there
exists at least one ball (among the candidates for the intersection)  that does not contain $\boldsymbol x$. Thus, the depth of the point $\boldsymbol x$ is $k-1$.

With a very small modification, the same algorithm can be applied to a point $\boldsymbol x$ out of the sample. We can just treat $\boldsymbol x$ as an artificial point,
in the same way as in previous sections. That is, the size of the required balls has to be $n-k+1$ points from the sample, not counting $\boldsymbol x$. The rest of
the algorithm is the same as in the case of a sample point $\boldsymbol x$.

In both versions (sample or out-of-sample) we can use additional artificial points to increase the precision. The sample version of the algorithm is
detailed below.

\begin{algorithm}[H]
  \LinesNumbered
  \SetAlgoLined
  \KwData{Original data, $X_n = (\boldsymbol{x_1}, \boldsymbol{x_1},...,\boldsymbol{x_n}) \in \mathbb{R}^{d \times n}$, $\boldsymbol x=\boldsymbol x_i$ for a fixed $i$ - the data point whose depth is calculated.}
  \KwResult{Tukey depth at $\boldsymbol x$.}
  \BlankLine
   \tcc{Iteration Phase}
   $S_{\alpha_1} = \{\boldsymbol x_1,\ldots,\boldsymbol x_n\}$\;
   \For {$k \gets 2$ \textbf{to} $n$} {

  	$p=0$ - Number of balls that contain $\boldsymbol x$. Its initial value is $0$ \;
    	\tcc{Find balls that contain point $\boldsymbol x$}
   	\For {$i \gets 1$ \textbf{to} $n$} {
   		\If {$\boldsymbol x \in B_i$, where $B_i$ contains $n - k + 1$ original data points} {
			$p=p+1$;
		}
   	}
	\If {$p \neq n$} {
		\Return $k-1$	
	}

   }
  \caption{Calculating Tukey depth of a sample point.}
\end{algorithm}

\subsubsection{Complexity}
\begin{thm}
\label{complexitytd}
Adapted ABCDepth algorithm for finding approximate Tukey depth of a sample point has order of $O(dn^2 + n^2\log{n})$ time complexity.
\end{thm}
\begin{proof}
Balls construction for Algorithm 2 is the same as in Algorithm 1 (lines 1-10), so by Theorem \ref{complexitytm}
this part runs in  $O(dn^2 + n^2 \log n)$ time.
For the point with the depth $\alpha_k$ algorithm enters in iteration loop $k$ times and it iterates through all $n$ points to
find the balls that contain $\boldsymbol x$ point, so the whole iteration phase runs in $O(kn)$ time.

Overall complexity of the  Algorithm 2 is:
\begin{equation}
\label{alg2}
O(dn^2 + n^2 \log n) + O(kn) \sim O(dn^2 + n^2 \log n).
\end{equation}
\end{proof}

\begin{rem}  When the input data set is sparse or when the sample set is small, we add artificial data to the original data set in order to
improve the algorithm accuracy. In that case, $n$ in (\ref{alg2}) should be replaced with $N$.
\end{rem}

\subsubsection{Examples}

To illustrate the output for the Algorithm 2, we use the same real data sets as we used for Algorithm 1.
For all data sets we applied Algorithm 2 in two runs; first time with sample points only and second time with additional
 $1000$ artificial points generated from uniform distribution. Points depths are verified using \textit{depth} function from
 \cite{struyf98} implemented in \cite{depthr}. For each data set we calculate the accuracy as $\frac{100k}{n}\%$, where $n$ is the sample size and
 $k$ is the number of points that has the correct depth compared with algorithm presented in \cite{struyf98}.

On Figure \ref{ny_depths_sample} we showed NY crime points depths with accuracy of $26\%$, but if we add more points to
the original data set as we showed on Figure \ref{ny_depths_artificial}, the accuracy is greatly improved to $87\%$.

\begin{figure}
\centering
\centering
\includegraphics[width=0.7\textwidth]{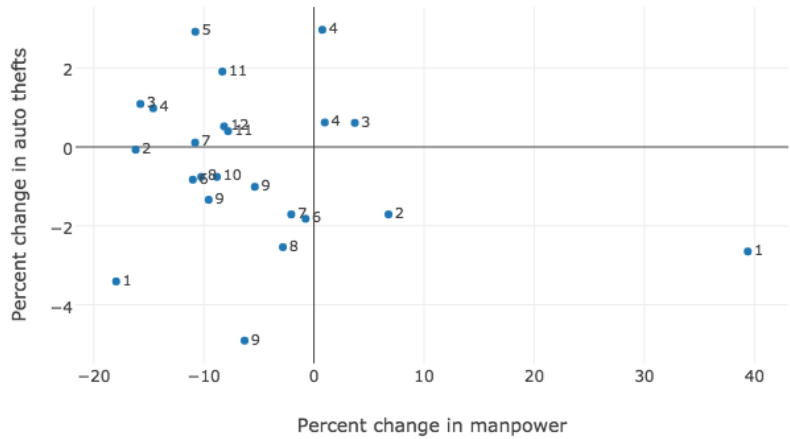}
\caption{NY crime data - point depths using only original data.}
\label{ny_depths_sample}
\end{figure}
\begin{figure}
\centering
\includegraphics[width=0.7\textwidth]{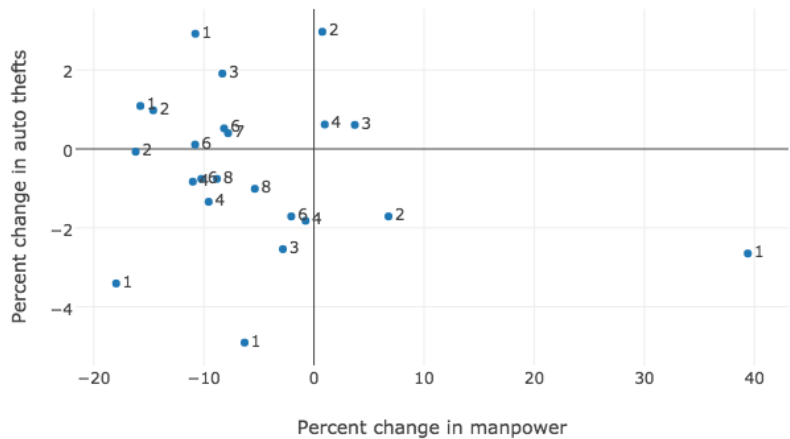}
\caption{NY crime data - point depths using original  and artificial data.}
\label{ny_depths_artificial}
\end{figure}

Figure \ref{animals_depth_sample} shows the same accuracy
of $26\%$ for animals data set, in the case when Algorithm 2 is run with sample points only.
 By adding more points as in Figure \ref{animals_depth_artificial}, the accuracy is improved to $92\%$.

\begin{figure}
\centering
\centering
\includegraphics[width=0.7\textwidth]{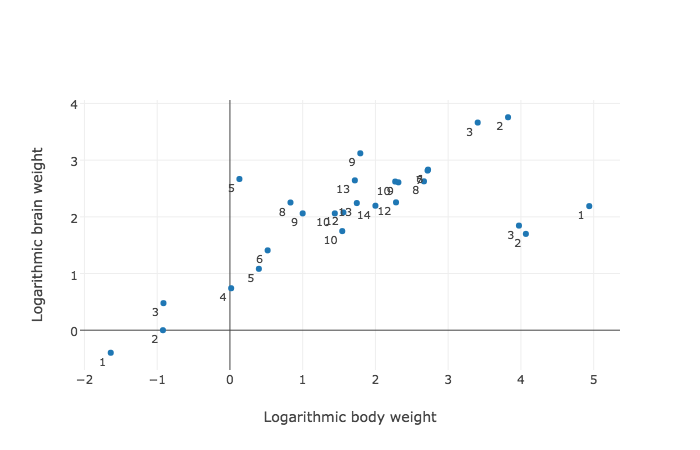}
\caption{Animals data - point depths using only original data.}
\label{animals_depth_sample}
\end{figure}
\begin{figure}
\centering
\includegraphics[width=0.7\textwidth]{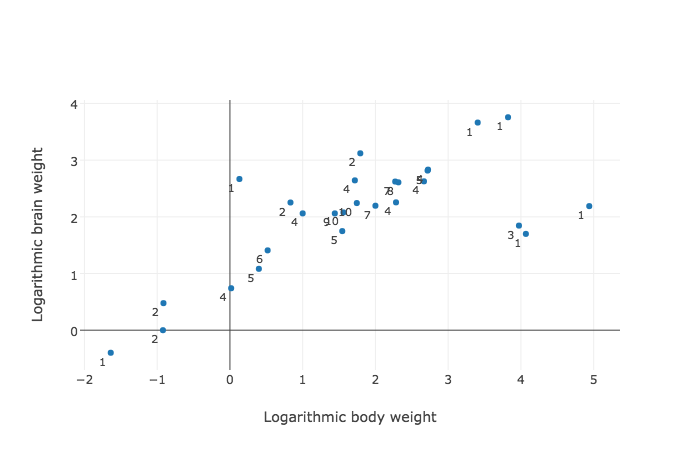}
\caption{Animals data -  point depths using  original and artificial data.}
\label{animals_depth_artificial}
\end{figure}

The third example is aircraft data set presented on Figure \ref{aircraft_depth_sample} and Figure \ref{aircraft_depth_artificial}.
The accuracy with artificial points is $95\%$, otherwise it is  $18\%$.

\begin{figure}
\centering
\centering
\includegraphics[width=0.7\textwidth]{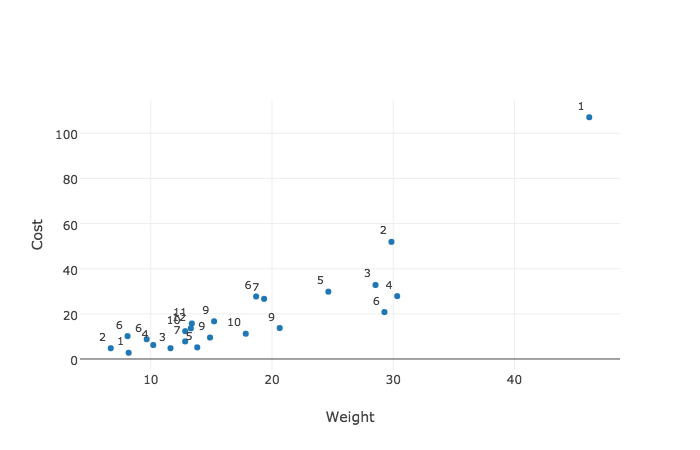}
\caption{Aircraft data - point depths using only original data.}
\label{aircraft_depth_sample}
\end{figure}
\begin{figure}
\centering
\includegraphics[width=0.7\textwidth]{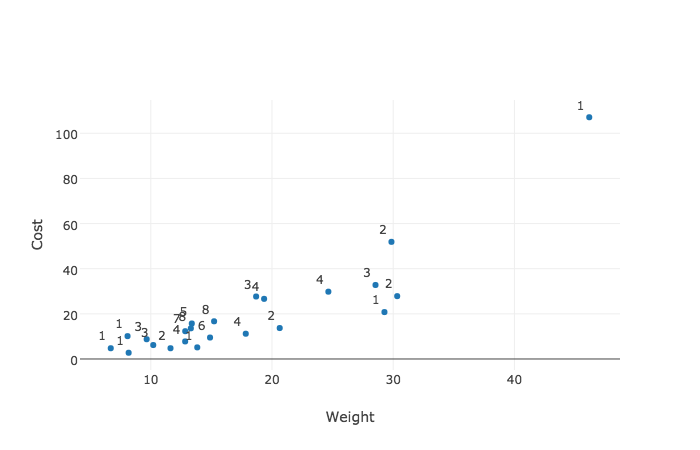}
\caption{Aircraft data - point depths using original  and artificial data.}
\label{aircraft_depth_artificial}
\end{figure}

As the last example of this section, we would like to calculate depths of the points plotted on Figure
\ref {single_ny} using ABCDepth Algorithm 2.  On Figure \ref {single_ny} we plotted Tukey median for NY crime data
set using Algorithm 1 with artificial data points (green point) and compared the result with the median obtained by
DEEPLOC (red point). Both points are out of the sample. On Figure \ref{ny_medians_depths} we show depths of all sample points
including the depths of two median points all attained by ABCDepth Algorithm 2. Algorithm presented in \cite{struyf98} and ABCDepth Algorithm 2 calculate the same depth value for both median points.

\begin{figure}
\centering
\includegraphics[width=0.7\textwidth]{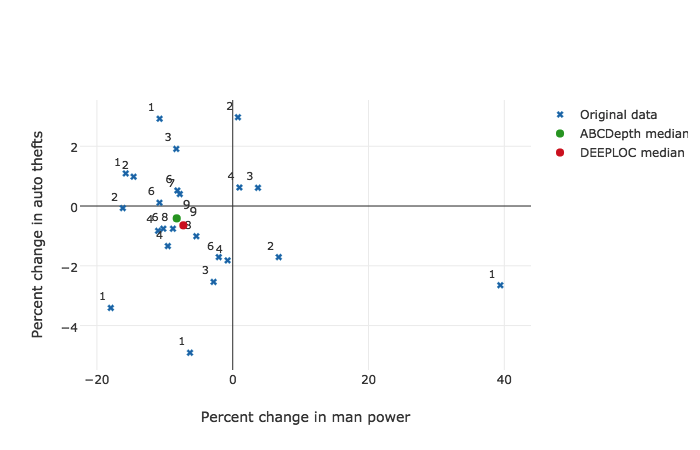}
\caption{Aircraft data - point depths using  original  and artificial data.}
\label{ny_medians_depths}
\end{figure}


\subsection{Adapted Implementation: finding depth contours}
\label{subsecdc}
Based on level sets $S = \{S_{\alpha_1}, S_{\alpha_2},...,S_{\alpha_m} \}$, one can obtain data depth contours applying \textit{QuickHull} algorithm implemented in \cite{quickhull96} on each level set.

In this purpose, to the original data set, $X_n=\{\boldsymbol{x_1}, \boldsymbol{x_2},...,\boldsymbol{x_n}\}$, we add artificial data points generated from the uniform distribution, $X'=\{\boldsymbol{x_{n_1}},...,\boldsymbol{x_N}\}$, so we denote the input data set as $X_n=\{\boldsymbol{x_1}, \boldsymbol{x_2},...,\boldsymbol{x_N}\}$.

As in Algorithm 2, level set $S_{\alpha_k}$ and consequently its depth contour $D_{\alpha_k}$ contain points with depth $\frac{k}{n}$, where $k=1,...,m$ and $S_{\alpha_m}$ is the deepest level set, i.e. $D_{\alpha_m}$ is corresponding deepest contour. Thus, in $k$th iteration each ball contains $n - k + 1$ points from the original data set, although the algorithm constructs and intersects all $N$ balls. This algorithm differs from Algorithm 1 and Algorithm 2 only in iteration phase. In addition, in lines 15-18 of the Algorithm 3, convex hulls for each level set can be calculated using \textit{QuickHull} algorithm \cite{quickhull96}.

\begin{rem}
{\rm
From Remark \ref{kmaxgp} it follows that the true number of contours is never smaller than the one produced by our algorithm.
}
\end{rem}

\begin{algorithm}[H]
  \LinesNumbered
  \SetAlgoLined
  \KwData{$X_N = (\boldsymbol{x_1}, \boldsymbol{x_2},...,\boldsymbol{x_n},\boldsymbol{x_{n+1}},...,\boldsymbol{x_N}) \in \mathbb{R}^{d \times N}$, where $x_i$ belongs to the original data set for $1 \leq i \leq n$ and $x_i$ belongs to the artificial data set for $n + 1 \leq i \leq N$

  $l$ - criterion for adding artificial data

  $p$ - number of artificial data points to add
  }
  \KwResult{List of level sets, $S = \{S_{\alpha_1}, S_{\alpha_2},...,S_{\alpha_m} \}$, set of depth contours $D=\{D_{\alpha_1}, D_{\alpha_2},...,D_{\alpha_m}\}$ }
  \BlankLine
   \tcc{Iteration Phase}
   $size = N$, $k=1$
   \While{$size > 1$} {
   	$S_{\alpha_k} = \{ \bigcap_{j}^{N} B_j,  \left | B_j \right |_n = n - k + 1$ - the size w.r.t. original points only \} \;
	$size = \left | S_{\alpha_k} \right |$ \;
	\If {$size < l$} {
		Generate $p$ artificial data points located in the region of $S_{\alpha_k}$ \;
		Add $p$ data points to the input set $X_N$ \;
		$N = N + p$ \;
		Repeat lines 1-6 from Algorithm 1 for new $N$\;
		Repeat lines 7-10 from Algorithm 1 for new $N$ \;
		$S_{\alpha_k} = \{ \bigcap_{j}^{N} B_j,  \left | B_j \right | = n - k + 1$ - the size w.r.t. original points only \} \;
	}
	$k = k + 1$ \;
	Add $S_{\alpha_k}$ to $S$ \;
   }
   \BlankLine
   \For {$i \gets 1$ \textbf{to} $m$} {
   	Calculate convex hull, $D_{\alpha_i}$, from level set $S_{\alpha_i}$ \;
	Add $D_{\alpha_i}$ to $D$ \;
   }
  \caption{Calculating depth contours.}
\end{algorithm}

Whenever the level set $S_{\alpha_k}$ contains less than $l$ data points, the algorithm adds $p$ artificial points to the input data set at the region of $S_{\alpha_k}$ in such a way that $S_{\alpha_k}$ is located centrally with respect to the additional artificial data. After that, the algorithm repeats the whole procedure for constructing balls for the augmented data set. The user can define the values of $l$ and $p$.

\subsubsection{Complexity discussion }
\label{complexitydc}

Level sets calculation has complexity which is linear in $d$ in all three ABCDepth algorithms.
That is the consequence of the fact that the number of dimensions plays a role only in Euclidian inter-distances
calculation (see lines 1-6 of the Algorithm 1).
Based on level sets produced in lines 1-15 of Algorithm 3, for each level set $S_{\alpha_k}$, where $k=1,...,m$,
the algorithm calculates corresponding depth contour $D_{\alpha_k}$ using \textit{QuickHull} algorithm \cite{quickhull96}.
According to \textit{QuickHull} algorithm, in $d \leq 3$ it runs in $O(n \log r)$ time, where $n$ is the number of input set points,
$r$ is the number of processed points and it is proportional to the number of vertices in the output.
Hence, the complexity ABCDepth algorithm for constructing depth contours in $d \leq 3$ is linear in $d$.
For $d \geq 4$ the complexity of \textit{QuickHull} grows exponentially with the number of dimensions
we will not cover that case in this paper.



\subsubsection{Examples}

As a demonstration of Algorithm 3, we consider real data sets: NY crime data,
data from \cite{animals} and \cite{aircraft}. For each data set we run \textit{isodepth} function based on
ISODEPTH algorithm \cite{rouru96b} implemented in \cite{depthr} and compare outputs.

Figure \ref{ny_thier_contours} and Figure \ref{ny_our_contours} present contours obtained by
ISODEPTH algorithm and Algorithm 3, respectively. One can note that both figures has the same number of contours,
i.e. the maximal depth is $\frac{10}{23}$, although DEEPLOC yields $\frac{9}{23}$ as the approximate maximal depth.
Contoure $D_2$ on Figure \ref{ny_our_contours} contains a point $(0.76, 2.97)$ which obviously doesn't belong to the depth
contour $D_2$ since its depth is $\frac{1}{23}$.

\begin{figure}
\centering
\begin{minipage}{0.45\textwidth}
\centering
\includegraphics[width=0.9\textwidth]{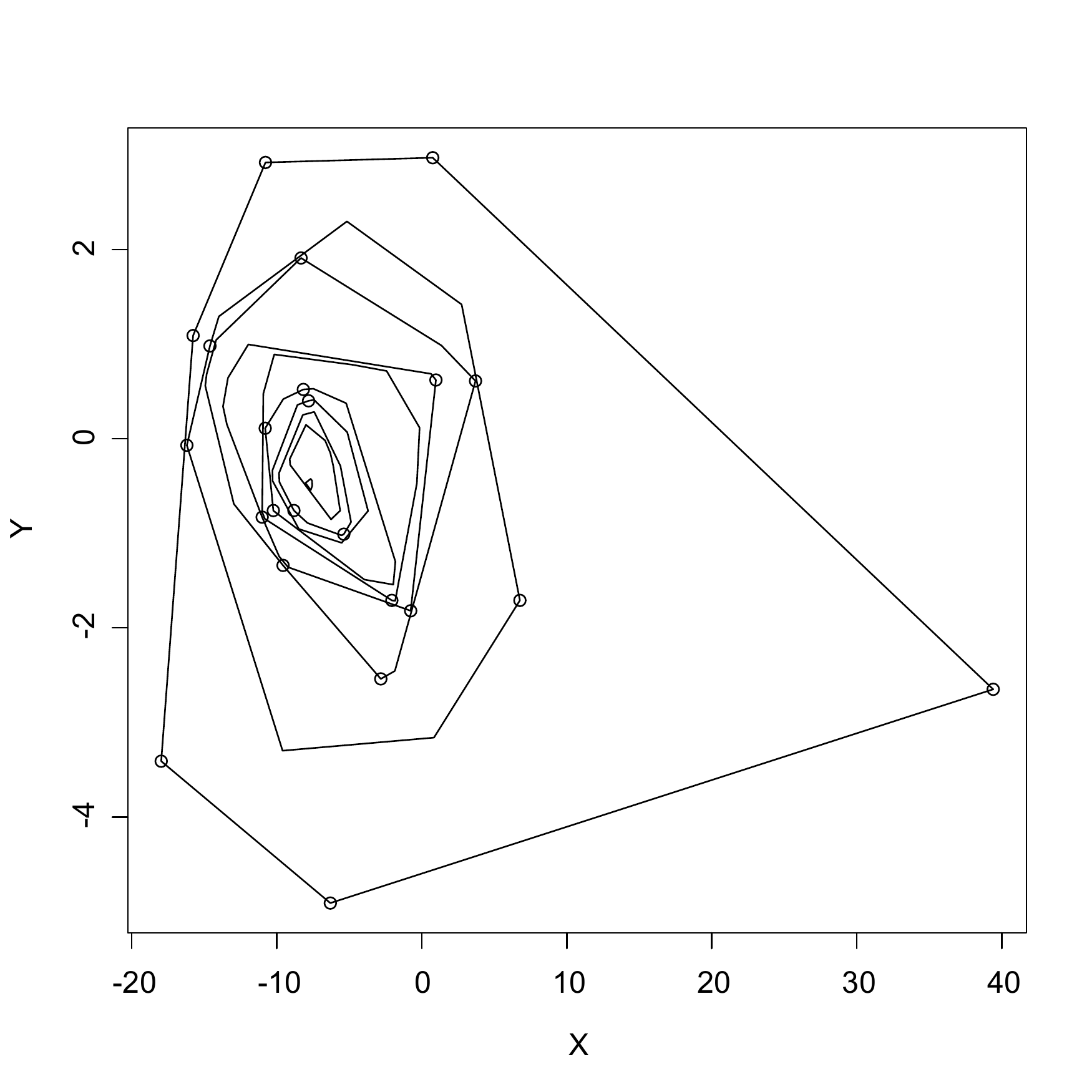}
\caption{NY crime data - ISODEPTH contours.}
\label{ny_thier_contours}
\end{minipage}\hfill
\begin{minipage}{0.5\textwidth}
\centering
\includegraphics[width=1.0\textwidth]{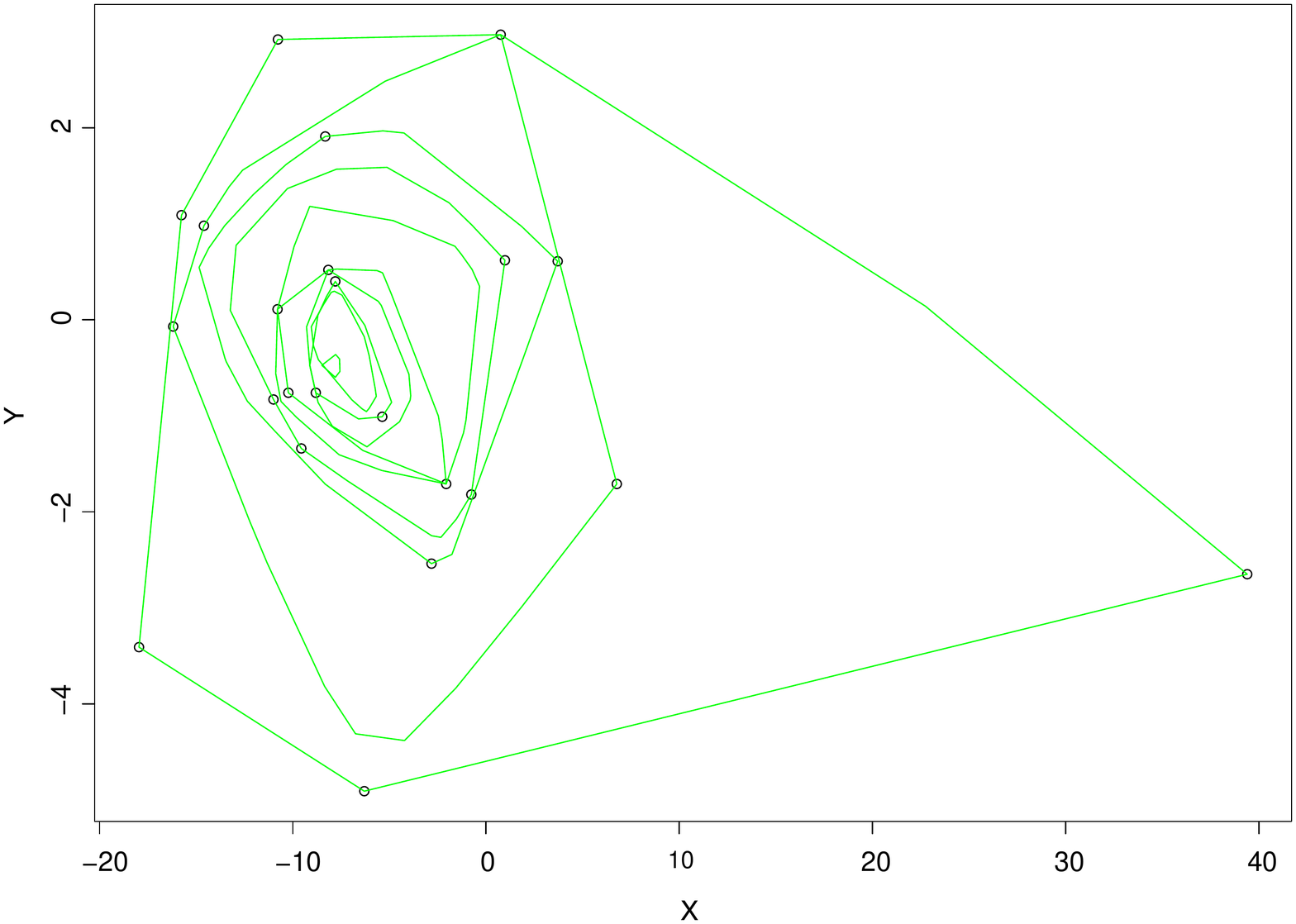}
\caption{NY crime data - ABCDepth contours.}
\label{ny_our_contours}
\end{minipage}
\end{figure}

Contours for animals data set are shown for both algorithms on Figure \ref{animals_thier_contours} and
Figure \ref{animals_our_contours}. ISODEPTH algorithm produces $11$ contours, although Figure 7 in \cite{rouru96b}
has $10$ contours. Algorithm 3 finds $10$ contours as well, i.e. the maximal depth is $\frac{10}{27}$. Aproximate maximal depth
calculated by DEEPLOC algorithm is $\frac{12}{27}$.

\begin{figure}
\centering
\begin{minipage}{0.45\textwidth}
\centering
\includegraphics[width=0.9\textwidth]{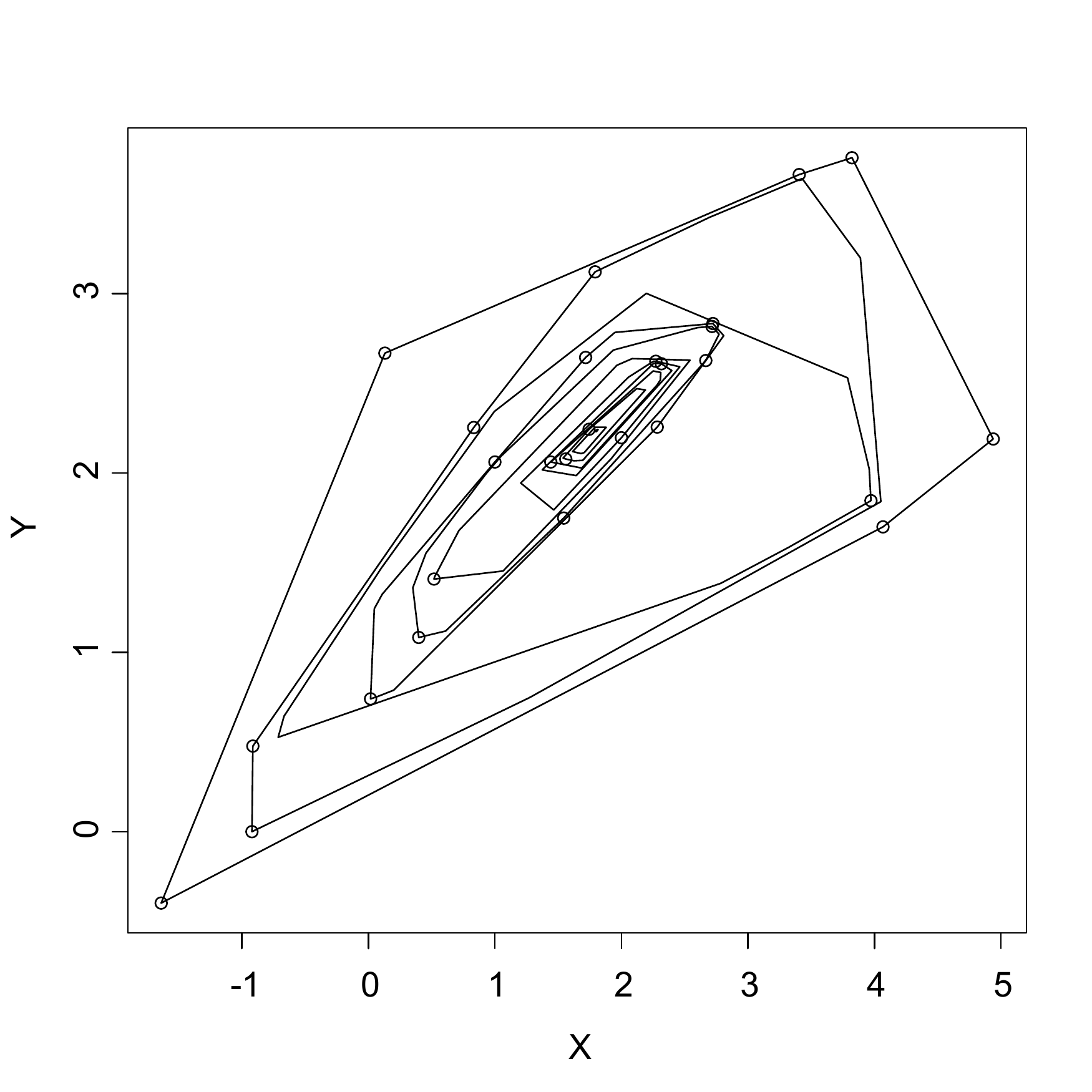}
\caption{Animals data - ISODEPTH contours.}
\label{animals_thier_contours}
\end{minipage}\hfill
\begin{minipage}{0.5\textwidth}
\centering
\includegraphics[width=1.0\textwidth]{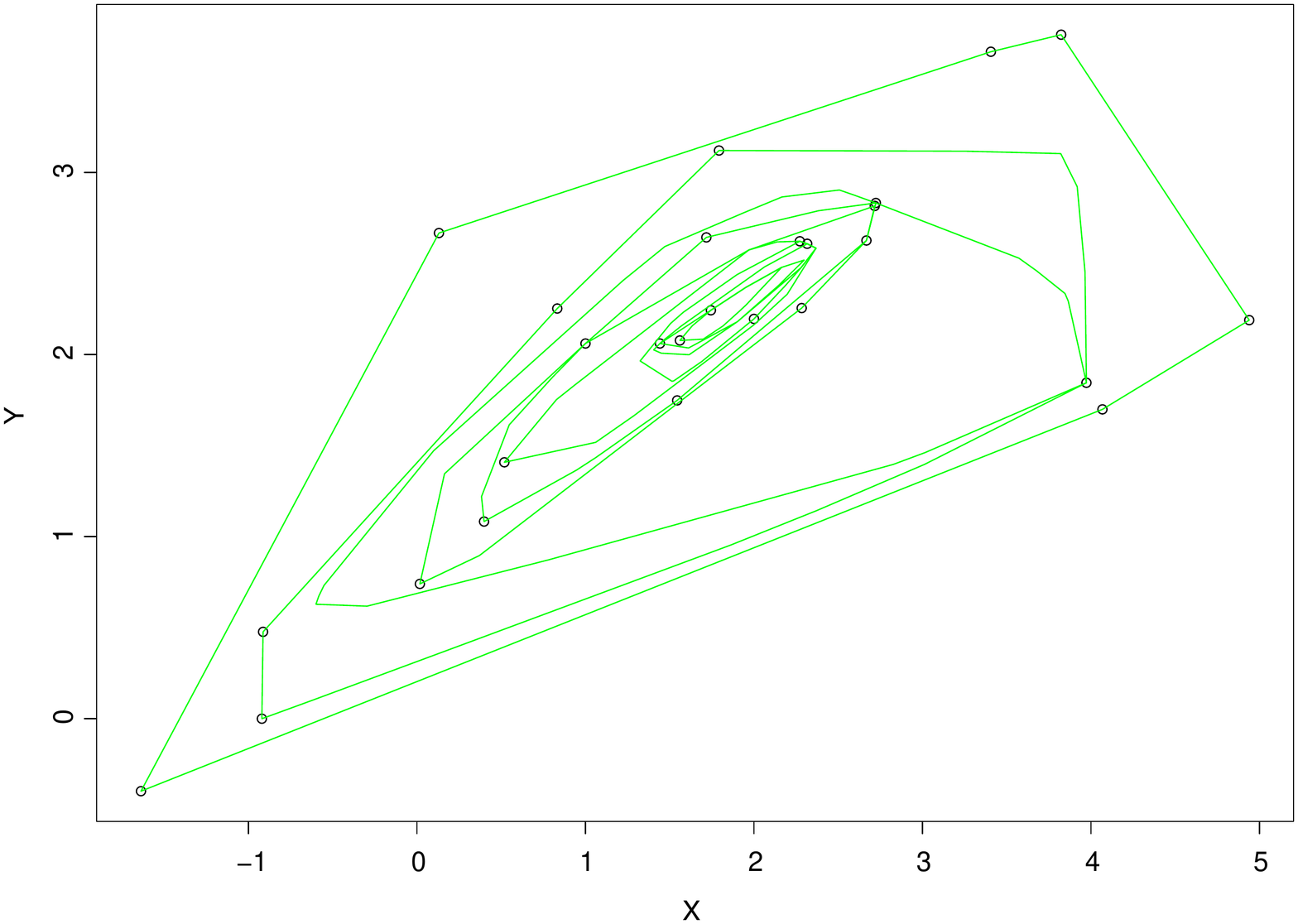}
\caption{Animals data - ABCDepth contours.}
\label{animals_our_contours}
\end{minipage}
\end{figure}

Figure \ref{aircraft_thier_contours} and Figure \ref{aircraft_our_contours} present
the depth contours for aircraft data set for both algorithms, ISODEPTH and Algorithm 3.
Each plot contains $10$ contours, i.e. the maximal depth is $\frac{10}{23}$ and DEEPLOC
finds the deepest point on the same depth.

\begin{figure}
\centering
\begin{minipage}{0.45\textwidth}
\centering
\includegraphics[width=0.9\textwidth]{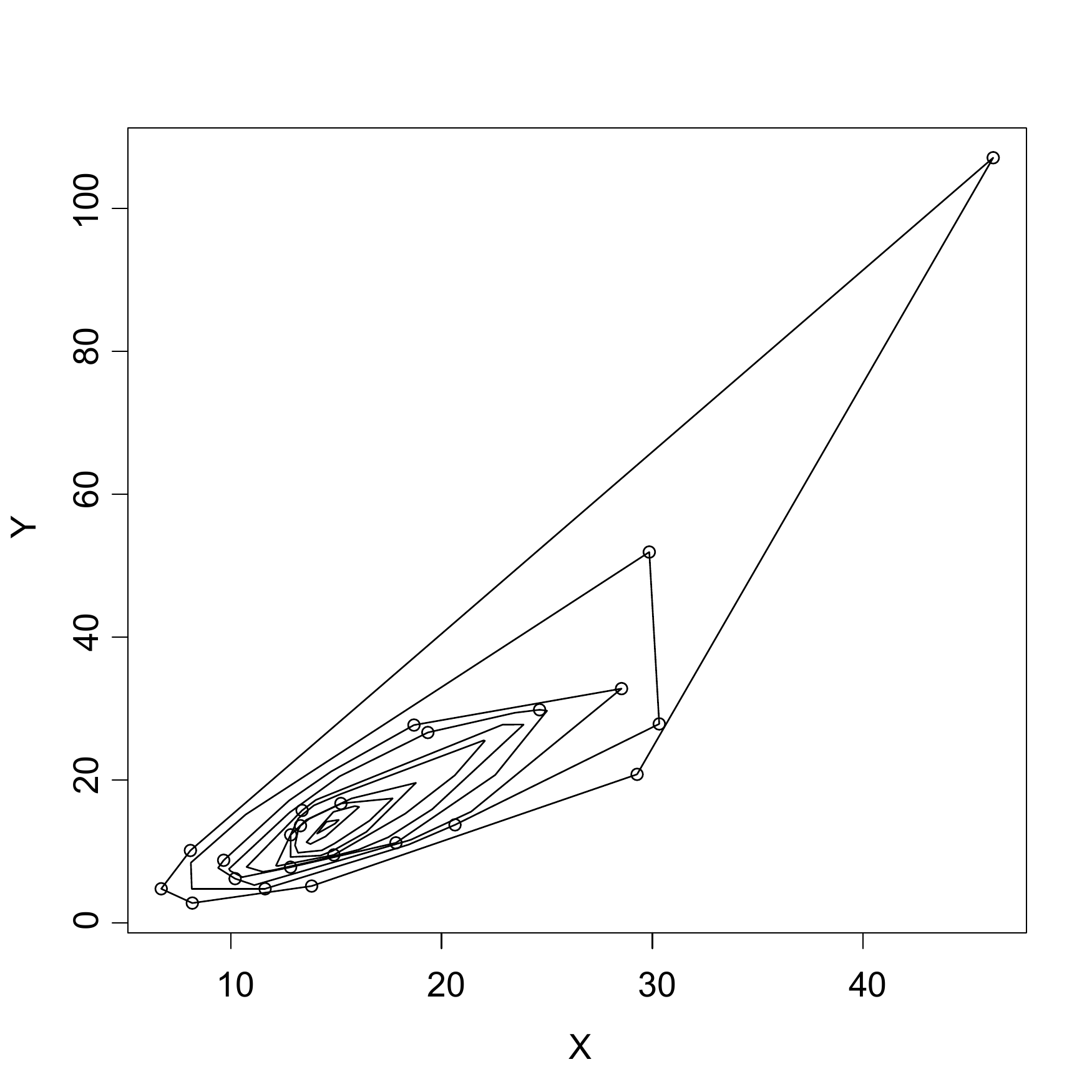}
\caption{Aircraft data - ISODEPTH contours.}
\label{aircraft_thier_contours}
\end{minipage}\hfill
\begin{minipage}{0.5\textwidth}
\centering
\includegraphics[width=1.0\textwidth]{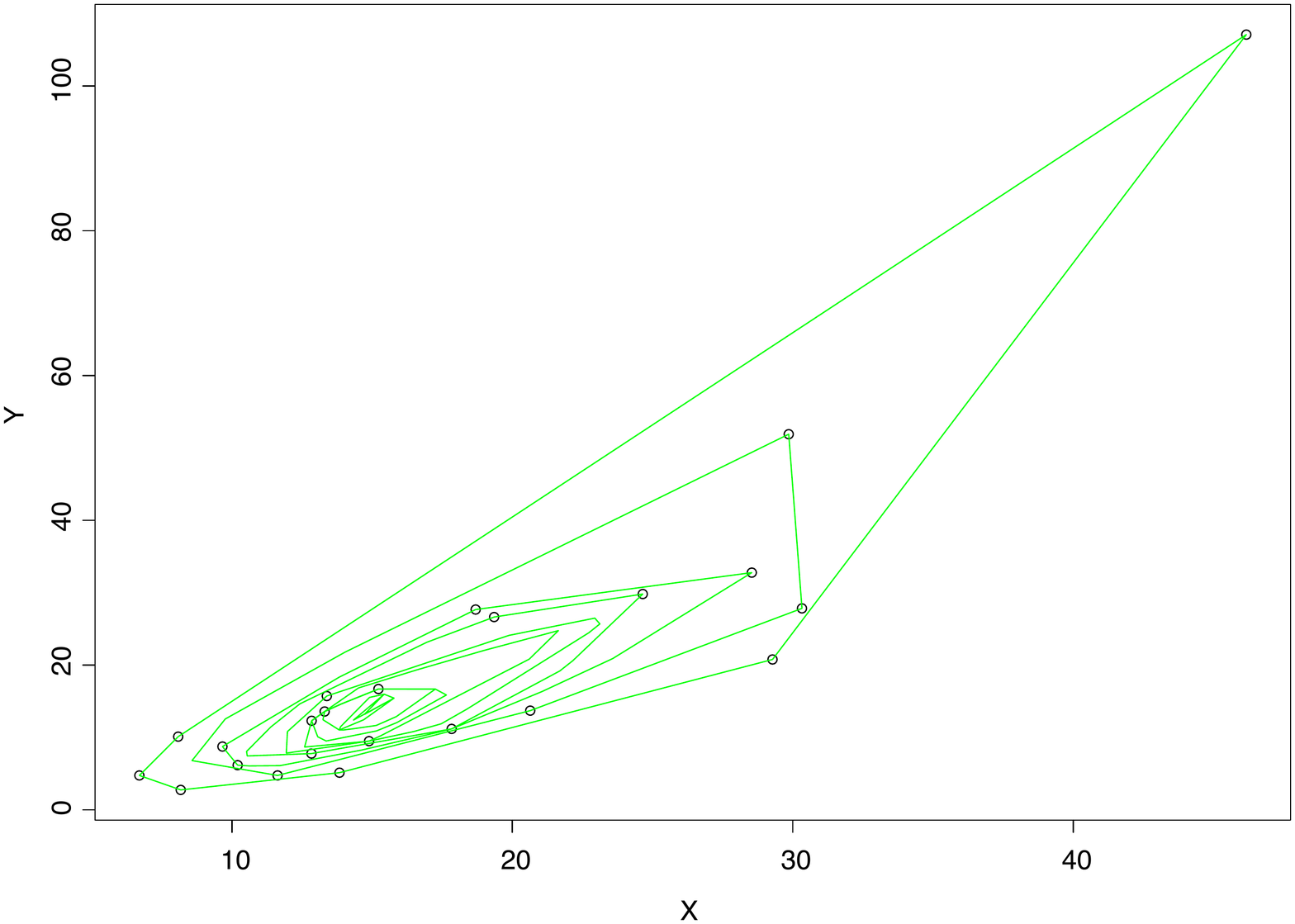}
\caption{Aircraft data - ABCDepth contours.}
\label{aircraft_our_contours}
\end{minipage}
\end{figure}

The contours produced by Algorithm 3 are similar to the contours
attained by ISODEPTH algorithm for all data sets we tested, but one should keep
in mind that Algorithm 3 contours are approximate and they result from intersections of balls,
which in real cases can not end with straight lines with finally many balls. In most cases, a depth contour
 obtained by Algorithm 3 contains the exact depth contour.

\section{Performance and Comparisons}
\label{secc}

According to Theorem \ref{complexitytm}, the complexity for calculating Tukey median grows linearly with dimension and in terms of a number of data points,
it grows with the order of $n^2\log n$.
Rousseeuw and Ruts in \cite{rouru98} pioneered with an exact algorithm called HALFMED for Tukey median in two dimensions
that runs in $O(n^2 log^2{n})$ time. This algorithm is better than ABCDepth for $d=2$, but it processes only bivariate data sets.
Struyf and Rousseeuw in \cite{struro00} implemented the first approximate algorithm called DEEPLOC for finding the deepest location
in higher dimensions. Its complexity is $O(kmn \log (n+kdn+md^3+mdn))$ time, where $k$ is the number of steps taken by the program and
m is the number of directions, i.e. vectors constructed by the program. This algorithm is very efficient for low-dimensional data sets,
but for high-dimensional data sets ABCDepth algorithm outperforms DEEPLOC.
Chan in \cite{chan04} presents an approximate randomized algorithm for maximum Tukey depth. It runs in $O(n^{d-1})$ time and it is
not implemented yet.

In Table 1 execution times of DEEPLOC algorithm and ABCDepth algorithm for finding Tukey median are reported.
The measurements are performed using synthetic data generated from the multivariate normal $\mathcal N(0,1)$ distribution.
In this table, we demonstrate how ABCDepth algorithm behaves with thousands of high-dimensional data points. It takes
$\sim 13$ minutes for $n=7000$ and $d=2000$. Since DEEPLOC algorithm doesn't support data sets with $d > n$ and returns the error message:
"the dimension should be at most the number of objects", we denoted those examples with $-$ sign in the table.
The sign $*$ means that the median is not computable at least once in $12$ hours.

\begin{table}[]
\centering
\caption{Compare DEEPLOC and ABCDepth execution times in seconds.}
\label{deepABC}
\begin{center}
\resizebox{0.6\textwidth}{!}{\begin{minipage}{\textwidth}
\begin{tabular}{@{}lllllllllllllll@{}}
\toprule
\multirow{2}{*}{d} & \multirow{2}{*}{Algorithm}                                 & \multicolumn{13}{c}{n}                                                                                                                                                                                                                                                                                                                                                                                                                                                                                                                                                                                                                                                                                                                                                    \\ \cmidrule(l){3-15}
                   &                                                            & 320                                                  & 640                                                     & 1280                                                 & 2560                                                 & 3000                                                 & 3500                                                   & 4000                                                  & 4500                                                   & 5000                                                  & 5500                                                   & 6000                                                   & 6500                                                   & 7000                                                                   \\ \midrule
50                 & \begin{tabular}[c]{@{}l@{}}Deeploc\\ ABCDepth\end{tabular} & \begin{tabular}[c]{@{}l@{}}4.43\\ 0.15\end{tabular}  & \begin{tabular}[c]{@{}l@{}}7.15\\ 0.63\end{tabular}     & \begin{tabular}[c]{@{}l@{}}12.65\\ 2.86\end{tabular} & \begin{tabular}[c]{@{}l@{}}23.87\\ 4.95\end{tabular} & \begin{tabular}[c]{@{}l@{}}30.93\\ 7.27\end{tabular} & \begin{tabular}[c]{@{}l@{}}31.79\\ 8.65\end{tabular}   & \begin{tabular}[c]{@{}l@{}}37.66\\ 12.51\end{tabular} & \begin{tabular}[c]{@{}l@{}}45.35\\ 14.18\end{tabular}  & \begin{tabular}[c]{@{}l@{}}50.72\\ 17.51\end{tabular} & \begin{tabular}[c]{@{}l@{}}63.13\\ 22.18\end{tabular}  & \begin{tabular}[c]{@{}l@{}}63.75\\ 25.86\end{tabular}  & \begin{tabular}[c]{@{}l@{}}84.13\\ 29.24\end{tabular}  & \begin{tabular}[c]{@{}l@{}}69.61\\ 37.34\end{tabular}                  \\ \midrule
100                & \begin{tabular}[c]{@{}l@{}}Deeploc\\ ABCDepth\end{tabular} & \begin{tabular}[c]{@{}l@{}}19.42\\ 0.22\end{tabular} & \begin{tabular}[c]{@{}l@{}}22.85\\ 0.92\end{tabular}    & \begin{tabular}[c]{@{}l@{}}33.81\\ 2.03\end{tabular} & \begin{tabular}[c]{@{}l@{}}77.45\\ 7.83\end{tabular} & \begin{tabular}[c]{@{}l@{}}69.04\\ 9.78\end{tabular} & \begin{tabular}[c]{@{}l@{}}105.56\\ 13.14\end{tabular} & \begin{tabular}[c]{@{}l@{}}97.39\\ 17.89\end{tabular} & \begin{tabular}[c]{@{}l@{}}120.05\\ 23.52\end{tabular} & \begin{tabular}[c]{@{}l@{}}140.04\\ 30.6\end{tabular} & \begin{tabular}[c]{@{}l@{}}131.85\\ 39.18\end{tabular} & \begin{tabular}[c]{@{}l@{}}127.36\\ 49.03\end{tabular} & \begin{tabular}[c]{@{}l@{}}212.42\\ 68.46\end{tabular} & \begin{tabular}[c]{@{}l@{}}183.27\\ 82.02\end{tabular}                 \\ \midrule
500                & \begin{tabular}[c]{@{}l@{}}Deeploc\\ ABCDepth\end{tabular} & \begin{tabular}[c]{@{}l@{}}-\\ 0.693\end{tabular}    & \begin{tabular}[c]{@{}l@{}}1616.53\\ 3.181\end{tabular} & \begin{tabular}[c]{@{}l@{}}*\\ 8.4\end{tabular}      & \begin{tabular}[c]{@{}l@{}}*\\ 27.9\end{tabular}     & \begin{tabular}[c]{@{}l@{}}*\\ 41.61\end{tabular}    & \begin{tabular}[c]{@{}l@{}}*\\ 53.73\end{tabular}      & \begin{tabular}[c]{@{}l@{}}*\\ 71.95\end{tabular}     & \begin{tabular}[c]{@{}l@{}}*\\ 89.36\end{tabular}      & \begin{tabular}[c]{@{}l@{}}*\\ 109.22\end{tabular}    & \begin{tabular}[c]{@{}l@{}}*\\ 140.18\end{tabular}     & \begin{tabular}[c]{@{}l@{}}*\\ 151.45\end{tabular}     & \begin{tabular}[c]{@{}l@{}}*\\ 180.5\end{tabular}      & \begin{tabular}[c]{@{}l@{}}*\\ 213.01\end{tabular}                     \\ \midrule
1000               & \begin{tabular}[c]{@{}l@{}}Deeploc\\ ABCDepth\end{tabular} & \begin{tabular}[c]{@{}l@{}}-\\ 1.165\end{tabular}    & \begin{tabular}[c]{@{}l@{}}-\\ 3.99\end{tabular}        & \begin{tabular}[c]{@{}l@{}}*\\ 14.389\end{tabular}   & \begin{tabular}[c]{@{}l@{}}*\\ 54.18\end{tabular}    & \begin{tabular}[c]{@{}l@{}}*\\ 74.38\end{tabular}    & \begin{tabular}[c]{@{}l@{}}*\\ 98.73\end{tabular}      & \begin{tabular}[c]{@{}l@{}}*\\ 129.85\end{tabular}    & \begin{tabular}[c]{@{}l@{}}*\\ 164.96\end{tabular}     & \begin{tabular}[c]{@{}l@{}}*\\ 203.37\end{tabular}    & \begin{tabular}[c]{@{}l@{}}*\\ 246.54\end{tabular}     & \begin{tabular}[c]{@{}l@{}}*\\ 286.17\end{tabular}     & \begin{tabular}[c]{@{}l@{}}*\\ 344.94\end{tabular}     & \begin{tabular}[c]{@{}l@{}}*\\ 39.16\end{tabular}                      \\ \midrule
2000               & \begin{tabular}[c]{@{}l@{}}Deeploc\\ ABCDepth\end{tabular} & \begin{tabular}[c]{@{}l@{}}-\\ 2.21\end{tabular}     & \begin{tabular}[c]{@{}l@{}}-\\ 7.86\end{tabular}        & \begin{tabular}[c]{@{}l@{}}-\\ 27.25\end{tabular}    & \begin{tabular}[c]{@{}l@{}}*\\ 107.46\end{tabular}   & \begin{tabular}[c]{@{}l@{}}*\\ 132.77\end{tabular}   & \begin{tabular}[c]{@{}l@{}}*\\ 180.02\end{tabular}     & \begin{tabular}[c]{@{}l@{}}*\\ 243.1\end{tabular}     & \begin{tabular}[c]{@{}l@{}}*\\ 297.6\end{tabular}      & \begin{tabular}[c]{@{}l@{}}*\\ 386.75\end{tabular}    & \begin{tabular}[c]{@{}l@{}}*\\ 475.87\end{tabular}     & \begin{tabular}[c]{@{}l@{}}*\\ 554.23\end{tabular}     & \begin{tabular}[c]{@{}l@{}}*\\ 666.4\end{tabular}      & \multicolumn{1}{c}{\begin{tabular}[c]{@{}c@{}}*\\ 764.74\end{tabular}} \\ \bottomrule
\end{tabular}
\end{minipage}}
\end{center}
\end{table}

ABCDepth algorithm for finding Tukey depth of a point runs in $O(dn^2 + n^2 \log n)$ as we showed in Theorem \ref{complexitytd}. Most of the algorithms for finding Tukey depth are exact and at the same time computationally expensive. One of the first exact algorithms for bivariate data sets, called LDEPTH, is proposed by Rousseeuw and Ruts in \cite{rouru96}. It has complexity of
$O(n \log n)$ and like HALFMED it outperforms ABCDepth for $d=2$.
Rousseeuw and Struyf in \cite{struyf98} implemented an exact algorithm for $d=3$ that runs in $O(n^2 \log n)$ time and an approximate algorithm for $d > 3$ that runs in $O(md^3 + mdn)$ where $m$ is the number directions, i.e. all directions perpendicular to hyperplanes through $d$ data points.
The later work of Chen et al. in \cite{chmowa13} presented approximate algorithms based on the third approximation method of Rousseeuw and Struyf in \cite{struyf98} reducing the problem from $d$ to $k$ dimensions. The first one, for $k=1$, runs in $O(\epsilon^{1-d}dn)$ time and the second one, for $k \geq 2$, runs in $O((\epsilon^{-1} c \log n)^d)$, where $\epsilon$ and $c$ are empirically chosen constants.
The another exact algorithm for finding Tukey depth in ${\bf R}^d$ is proposed by Liu and Zuo in \cite{zuo14}, which proves to be extremely time-consuming (see Table 5.1 of Section 5.3 in \cite{pavlo14}) and the algorithm involves heavy computations, but can serve as a benchmark.
Recently, Dyckerhoff and Mozharovskyi in \cite{pavlo16} proposed two exact algorithm for finding halfspace depth that run in $O(n^d)$ and $O(n^{d-1} \log n)$ time.

Table 2 shows execution times of ABCDepth algorithm for finding a depth of a sample point. Measurements are
derived from synthetics data from the multivariate standard normal distribution. Execution time for each data set represents
averaged time consumed per data point. Most of the execution time ($\sim 95\%$) is spent on balls construction (see lines 1-10 of the Algorithm 1),
while finding a point
depth itself (iteration phase of the Algorithm 2) is really fast since it runs in $O(kn)$ time.

\begin{table}[]
\centering
\caption{Aveerage time per data point.}
\label{depthABC}
\begin{center}
\resizebox{0.6\textwidth}{!}{\begin{minipage}{\textwidth}
\begin{tabular}{@{}llllllllllllll@{}}
\toprule
\multirow{2}{*}{d} & \multicolumn{13}{c}{n}                                                                                                      \\ \cmidrule(l){2-14}
                   & 320  & 640  & 1280 & 2560  & 3000  & 3500  & 4000  & 4500  & 5000   & 5500   & 6000   & 6500   & 7000                       \\ \midrule
50                 & 0.07 & 0.21 & 1.21 & 8.23  & 12.64 & 19.22 & 28.56 & 42.04 & 64.33  & 77.45  & 98.79  & 121.86 & 150.73                     \\ \midrule
100                & 0.08 & 0.25 & 1.23 & 8.18  & 13.91 & 20.48 & 28.51 & 44.31 & 65.55  & 81.91  & 99.96  & 123.84 & 154.65                     \\ \midrule
500                & 0.13 & 0.42 & 1.84 & 11.42 & 17.93 & 21.42 & 35.41 & 52.07 & 73.21  & 95.18  & 119.82 & 141.88 & 176.12                     \\ \midrule
1000               & 0.17 & 0.53 & 2.52 & 13.53 & 20.13 & 32.35 & 41.71 & 58.72 & 82.92  & 103.84 & 138.69 & 155.32 & 200.55                     \\ \midrule
2000               & 0.26 & 0.94 & 4.12 & 18.32 & 28.12 & 38.79 & 56.04 & 73.79 & 102.98 & 124.48 & 156.54 & 186.59 & \multicolumn{1}{c}{232.45} \\ \bottomrule
\end{tabular}\end{minipage}}
\end{center}
\end{table}

In Section \ref{subsecdc} we presented ABCDepth algorithm for calculating level sets and in addition it can construct depth contours using \textit{QuickHull} algorithm. Its complexity is linear in $d$ for $d \leq 3$. For $d=2$ there are two exact algorithms for constructing depth contours. The first one, called ISODEPTH, is proposed by Ruts nad Rousseeuw in \cite{rouru96b} and for $n < 1000$ the time of the proposed algorithm behaves as a multiple of $n^2 \log n$, although according to \textit{isodepth} function implemented in R "depth" package \cite{depthr} ISODEPTH takes several minutes to calculate contours from $1000$ points generated from bivariate normal distribution.
The second algorithm is presented by Miller et al. in \cite{struyf03} which computes all bivariate depth contours in $O(n^2)$ time.
For the depth contours in dimensions $d > 2$ Liu et al. proposed an algorithm in \cite{pavlo14a} that runs in $O(n^p \log n)$ time.

The ABCDepth algorithm has been implemented in Java. Tests for all algorithms are run using one kernel of Intel Core i7 (2.2 GHz) processor.

\section*{Acknowledgements}
We would like to express our gratitude to Anja Struyf and coauthors for sharing the code and the data
that were used in their papers of immense importance in the area. Answering to Yijun Zuo's doubts about the first arXiv version of this paper
and solving difficult queries that he was proposing, helped us to improve the
presentation and the algorithms. The second author acknowledges the  support by grants III 44006 and 174024
from Ministry of Education, Science and Technological Development of
Republic of Serbia.


\bibliography{btukey}

\begin{thebibliography}{10}

\bibitem{bome14}
M.~Bogicevic and M.~Merkle.
\newblock {Multivariate Medians and Halfspace Depth: Algorithms and
  Implementation}.
\newblock In {\em Proc. 1st International Conference on Electrical, Electronic
  and Computing Engineering (IcETRAN 2014), Vrnja\v cka Banja, Serbia},
  volume~1, page~27, 2014.

\bibitem{bome15}
M.~Bogicevic and M.~Merkle.
\newblock {Data Centrality Computation: Implementation and Complexity
  Calculation}.
\newblock In {\em Proc. 2nd International Conference on Electrical, Electronic
  and Computing Engineering (IcETRAN 2015), Srebrno Jezero, Serbia}, volume~1,
  page~23, 2015.

\bibitem{quickhull96}
Hannu~Huhdanpaa C.~Bradford~Barber, David P.~Dobkin.
\newblock {The quickhull algorithm for convex hulls}.
\newblock {\em ACM Transactions on Mathematical Software}, 22:469--483, 1996.

\bibitem{chan04}
T.~M. Chan.
\newblock {An Optimal Randomized Algorithm for Maximum Tukey Depth}.
\newblock In {\em Proceedings of the Fifteenth Annual ACM-SIAM Symposium on
  Discrete Algorithms}, pages 430--436. ACM, New York, 2004.

\bibitem{chmowa13}
Dan Chen, Pat Morin, and Uli Wagner.
\newblock {Absolute approximation of Tukey depth: Theory and experiments}.
\newblock {\em Comput. Geom.}, 46:566--573, 2013.

\bibitem{fastset}
Bolin Ding and Arnd~Christian König.
\newblock {A Fast set intersection in memory}.
\newblock {\em {Proceedings of the VLDB Endowment }}, 4:255--266, 2011.

\bibitem{doga92}
D.~L. Donoho and M.~Gasko.
\newblock {Breakdown properties of location estimates based on halfspace depth
  and projected outlyingness}.
\newblock {\em Ann. Statist.}, 20:1803--1827, 1992.

\bibitem{dughocha11}
S.~Dutta, A.~K. Ghosh, and P.~Chaudhuri.
\newblock {Some intriguing properties of Tukey's half-space depth}.
\newblock {\em Bernoulli}, 17:1420--1434, 2011.

\bibitem{pavlo16}
Rainer Dyckerhoff and Pavlo Mozharovskyi.
\newblock {Exact computation of the halfspace depth}.
\newblock {\em Computational Statistics and Data Analysis}, 98:19--30, 2016.

\bibitem{depthr}
Maxime Genest, Jean-Claude, and Jean-Francois Plante.
\newblock Package depth, 2012.

\bibitem{aircraft}
J.B. Gray.
\newblock {Graphics for regression diagnostics,}.
\newblock {\em ASA Proc. Statistical Computing Section}, pages 102--107, 1985.

\bibitem{hoare61}
C.~A.~R. Hoare.
\newblock {Algorithm 64: Quicksort}.
\newblock {\em Comm. Acm.}, 4:321, 1961.

\bibitem{pavlo14a}
X.~Liu, K.~Mosler, , and P.~Mozharovskyi.
\newblock {Fast computation of Tukey trimmed regions in dimension $p > 2$}.
\newblock {\em arXiv:1412.5122}, 2014.

\bibitem{zuo14}
X.~Liu and Y.~Zuo.
\newblock {Computing halfspace depth and regression depth}.
\newblock {\em Communications in Statistics Simulation and Computation},
  43:969--985, 2014.

\bibitem{merkle05}
Milan Merkle.
\newblock {Jensen's inequality for medians}.
\newblock {\em Stat. Prob. Letters}, 71:277--281, 2005.

\bibitem{merkle10}
Milan Merkle.
\newblock {Jensen's inequality for multivariate medians}.
\newblock {\em J. Math. Anal. Appl.}, 370:258--269, 2010.

\bibitem{struyf03}
Kim Miller, Suneeta Ramaswami, Peter Rousseeuw, J.Antoni Sellares, Diane
  Souvaine, Ileana Streinu, and Anja Struyf.
\newblock {Efficient computation of location depth contours by methods of
  computational geometry}.
\newblock {\em Statistics and Computing}, 13:153--162, 2003.

\bibitem{pavlo14}
Pavlo Mozharovskyi.
\newblock {\em Contributions to depth-based classification and computation of
  the Tukey depth}.
\newblock PhD thesis, Faculty of Economics and Social Sciences, University of
  Cologne, 2014.

\bibitem{animals}
P.~J. Rousseeuw and A.~M. Leroy.
\newblock {Robust Regression and Outlier Detection}.
\newblock {\em Wiley}, page~57, 1997.

\bibitem{rouru96}
P.~J. Rousseeuw and I.~Ruts.
\newblock {Bivariate Location Depth}.
\newblock {\em Journal of the Royal Statistical Society. Series C (Applied
  Statistics)}, 45:516--526, 1996.

\bibitem{rouru98}
P.~J. Rousseeuw and I.~Ruts.
\newblock {Constructing the Bivariate Tukey Median}.
\newblock {\em Constructing the Bivariate Tukey Median}, 8:827--839, 1998.

\bibitem{rouru99}
P.~J. Rousseeuw and I.~Ruts.
\newblock {The depth function of a population distribution}.
\newblock {\em Metrika}, 49:213--244, 1999.

\bibitem{survey15}
Peter~J. Rousseeuw and Mia Hubert.
\newblock {Statistical depth meets computational geometry: a short survey}.
\newblock {\em arXiv}, 2015.

\bibitem{struyf98}
Peter~J. Rousseeuw and Anja Struyf.
\newblock {Computing location depth and regression depth in higher dimension}.
\newblock {\em Statistics and Computing}, 8:193--203, 1998.

\bibitem{rouru96b}
I.~Ruts and P.~J. Rousseeuw.
\newblock {Computing depth contours of bivariate point clouds}.
\newblock {\em Computational Statistics and Data Analysis}, 23:153--168, 1996.

\bibitem{small90}
C.~G. Small.
\newblock {A survey of multidimensional medians}.
\newblock {\em arXiv}, 58:263--277, 1990.

\bibitem{struro00}
A.~Struyf and P.~J. Rousseeuw.
\newblock {High-dimensional computation of the deepest location}.
\newblock {\em Comp. Statist. \& Data Anal.}, 34:415--426, 2000.

\bibitem{tukey75}
John Tukey.
\newblock {Mathematics and Picturing Data}.
\newblock In {\em Proc. International Congress of Mathematicians, Vancouver
  1974}, volume~2, pages 523--531, 1975.

\bibitem{zhouserf08}
Y.~Zhou and R.~Serfling.
\newblock {Multivariate spatial U-quantiles: A Bahadur-Kiefer representation, a
  Theil-Sen estimator for multiple regression, and a robust dispersion
  estimator}.
\newblock {\em J. Statist. Plann. Inference}, 138:1660--1678, 2008.

\bibitem{zuoserf00}
Y.~Zuo and R.~Serfling.
\newblock {General notions of statistical depth function}.
\newblock {\em Ann. Stat.}, 28:461--482, 2000.

\end{thebibliography}

\bibliographystyle{plain}

\end{document}